\documentclass[sn-mathphys,Numbered]{sn-jnl}

\usepackage{graphicx}%
\usepackage{multirow}%
\usepackage{amsmath,amssymb,amsfonts}%
\usepackage{amsthm}%
\usepackage{mathrsfs}%
\usepackage[title]{appendix}%
\usepackage{xcolor}%
\usepackage{textcomp}%
\usepackage{manyfoot}%
\usepackage{booktabs}%
\usepackage{algorithm}%
\usepackage{algorithmicx}%
\usepackage{cleveref}
\usepackage{algpseudocode}%
\usepackage{tikz}
\usepackage{listings}
\usepackage{caption}
\usepackage{subcaption}

\newtheorem{lemma}{Lemma}

\newtheorem{corollary}{Corollary}
\newtheorem{observation}{Observation}

\theoremstyle{plain}%
\newtheorem{theorem}{Theorem}
\newtheorem{proposition}[theorem]{Proposition}%

\theoremstyle{plain}%

\theoremstyle{plain}%
\newtheorem{definition}{Definition}%

\newtheorem{protocol}[theorem]{Protocol}

\newcommand{\N}{\mathbb{N}}
\newcommand{\Z}{\mathbb{Z}}

\newcommand{\spanningtree}{\textsc{Spanning-Tree}}
\newcommand{\sizeofG}{\textsc{Size}}
\newcommand{\stpath}{s,t-\textsc{Path}}
\newcommand{\permutation}{\textsc{Permutation}}

\newcommand{\hide}[1]{ }

\newcommand{\trapezoid}{\textsc{Trapezoid}}

\newcommand{\chordal}{\textsc{chordal}}
\newcommand{\interval}{\textsc{interval}}
\newcommand{\circulararc}{\textsc{Circular-Arc}}
\newcommand{\propInterval}{\textsc{proper interval}}

\newcommand{\circularArc}{\textsc{circular-arc}}
\newcommand{\propCircular}{\textsc{proper circ-arc}}

\newcommand{\id}{\ensuremath{\mathsf{id}}}

\newcommand{\cO}{\mathcal{O}}

\newcommand{\tree}{\langle \id(\rho) , d_v , t_v\rangle}
\usepackage{enumitem}

\newcommand{\poly}{\mathrm{poly}}

\newcommand{\eps}{\varepsilon}

\begin{document}

\title[geometrics]{Local Certification of Some Geometric Intersection Graph Classes\footnote{This work was supported by Centro de Modelamiento Matemático (CMM), ACE210010 and FB210005,
BASAL funds for centers of excellence from ANID-Chile, FONDECYT 1230599 (P.M.), ANID-Subdirección de Capital Humano/Doctorado Nacional/2023 21230743 (B.J) and FONDECYT 1220142 (I.R.).}}


\author[1]{Benjamin Jauregui}\email{bjauregui@dim.uchile.cl}

\author*[2]{Pedro Montealegre}\email{p.montelaegre@uai.cl}

\author[1]{Diego Ramirez-Romero}\email{dramirez@dim.uchile.cl}

\author[1,2]{Iván Rapaport}\email{rapaport@dim.uchile.cl}

\affil*[1]{\orgdiv{Departamento de Ingeniería Matemática}, \orgname{Universidad de Chile}, \orgaddress{\street{Beauchef 851}, \city{Santiago}, \postcode{8370456}, \country{Chile}}}

\affil*[2]{\orgdiv{Facultad de Ingeniería y Ciencias}, \orgname{Universidad Adolfo Ibañez}, \orgaddress{\street{Av. Diagonal Las Torres 2640}, \city{Santiago}, \postcode{7941169}, \country{Chile}}}

\affil*[3]{\orgdiv{Centro de Modelamiento Matemático (UMI 2807 CNRS)}, \orgname{Universidad de Chile}, \orgaddress{\street{Beauchef 851}, \city{Santiago}, \postcode{8370456}, \country{Chile}}}


\abstract{In the context of distributed certification, the recognition of graph classes has started to be intensively studied. For instance, different results related to the recognition of planar, bounded tree-width and  $H$-minor free graphs have been recently obtained. The goal of the present work is to  design compact certificates for the local recognition of relevant geometric intersection graph classes, namely interval, chordal, circular arc, trapezoid and permutation.  More precisely, we give proof labeling schemes recognizing each of these classes with logarithmic-sized certificates. We also provide  tight logarithmic lower bounds on the size of the certificates on the proof labeling schemes for the recognition of any of the aforementioned  geometric intersection graph classes. }




\maketitle

\section{Introduction}\label{sec:intro} 

This paper examines the standard scenario of distributed network computing, where nodes in a network, represented as a graph $G=(V,E)$, exchange information through the links $E$ of the graph (see, for example, Peleg~\cite{Peleg00}). The objective is to gain a deeper understanding of the locality of graph properties. For instance, let's consider the property ``every node has an even number of neighbors." This property can be checked locally, meaning that if each node verifies that it has an even number of neighbors, then the graph satisfies the property.

Similar to centralized computing, distributed algorithms often make assumptions about the properties of $G$, and many algorithms are designed for specific types of graphs, such as regular graphs, planar graphs, bipartite graphs, or graphs with bounded tree-width. However, most graph properties of interest are not locally checkable. For instance, determining whether the graph has an even number of vertices requires nodes to examine beyond their immediate vicinity. Other natural properties like aciclicity or planarity requires the nodes to look arbitrarily far in the graph to verify them.

To cope with properties that are not locally checkable, several model extensions have been proposed. One possible solution is through local certification, which enables the local verification of any graph property. A local certification consists of a certificate assignment and a verification algorithm for a specific property. Together with the input information, each node receives a certificate and executes the verification algorithm communicating with their neighborhood. This algorithm determines whether the node accepts or rejects the certification. The protocol has to satisfy soundness and completeness conditions. Namely, if the graph satisfies the property, there exists a certificate assignment where all nodes accept it. Conversely, if a property is not satisfied, there is at least one node that rejects the certificate in every assignment.

In recent years, the field of local certification has gained considerable attention. We refer to Feuilloley~\cite{feuilloley2021introduction} for an introduction on the area. 

Proof-labeling schemes (PLSs) are, arguably, the best-known local certification type of protocol. They were introduced by Korman, Kutten, and Peleg in 2010 \cite{KormanKP10}. PLSs represent one of the weakest forms of local certification, where the verification algorithm is restricted to sharing the certificates in just one round of communication. In simpler terms, each node runs a verification algorithm with knowledge limited to its own certificate and the certificates of its neighbors in the graph. Despite these limitations, PLSs exhibit remarkable capabilities when compared to other local certification algorithms.

It is known that any property can be certified by a PLS using certificates of size \(\cO(n^2)\) bits, where \(n\)  is the total number of vertices. This can be achieved by providing each node with a complete description of the graph, allowing them to verify the property and the correctness of the local graph description. However, the \(\cO(n^2)\) certificate size is excessively large. Therefore, the primary objective in the study of local certification is to minimize the certificate size, expressed in terms of bits per vertex as a function of 
\(n\). Determining the minimum certificate size holds theoretical significance, as the optimal certificate size of a property reflects its locality: smaller certificates imply less dependence on global information, indicating a more localized property.

Motivated by the results of Göös and Suomela \cite{goos2016locally}, in \cite{FeuilloleyBP22} the authors remarked that $\Theta(\log n)$ is a benchmark for the number of bits that one can hope for a PLS to achieve. Indeed, certificates of size $o(\log n)$ are too short even for very simple properties. For instance, any PLS that verifies acyclicity requires certificates of $\Omega(\log n)$ bits \cite{KormanKP10}. On the other hand, 
a logarithmic number of bits allows us to encode identifiers, distances, spanning trees, etc. For these reasons, a certification with $\Theta(\log n)$ bits is called a \emph{compact local certification}.

Unfortunately, not every property has a compact certification. For example, not being 3-colorable cannot be certified with less than $\Omega(n^2/\log n)$ bits \cite{goos2016locally}. This is in sharp contrast with the problem of verifying 3-colorability: there is a trivial PLS to verify whether the graph is 3-colorable with two bits, which simply assigns each vertex a number in ${0,1,2}$ representing its color in a proper 3-coloring.

 In \cite{FeuilloleyBP22}, the authors raise the question of which graph properties admit compact certifications. In recent years, several results have emerged demonstrating that many relevant graph classes can be recognized using compact certificates (see the Related Work section below for more details). In this article, we study a specific set of graph properties defined by the intersection of geometric objects.

\subsection{Geometric Intersection Graph Classes}

A graph $G=(V,E)$ is a geometric intersection graph if every node $v \in V$ is identified with a geometric object of some particular type, and two vertices are adjacent if the corresponding objects intersect. 
Intersection graphs are the natural model of wireless sensor networks (where simple devices are deployed in large areas),  but they also appear in disciplines that do not necessarily come from distributed computing such as biology, ecology, matrix analysis, circuit design, statistics, archaeology, scheduling, etc. For a nice survey, we refer to~\cite{McKee1999}.

The two simplest non-trivial, and arguably two of the most studied
geometric intersection graphs are  {\emph{interval graphs}} and  {\emph{permutation graphs}}. In fact, most of the best-known geometric intersection graph classes are either generalizations of interval graphs or generalizations of permutation graphs.
It comes  as no surprise that many papers address different algorithmic and structural  aspects, simultaneously, in both interval and permutation graphs \cite{asdre2007harmonious,kante2013enumeration,kratsch2006certifying,yamazaki2020enumeration}.

In both interval and permutation graphs, the intersecting objects  are (line) segments, with different restrictions imposed on their positions.
In interval graphs, the segments must all lie on the real line. In permutation graphs, the endpoints of the segments must  lie on two separate, parallel real lines. In 
Figure~\ref{fig:intervalEx} we give an example of an interval graph, while in Figure~\ref{fig:Expermutation}  we give an example of an interval graph.

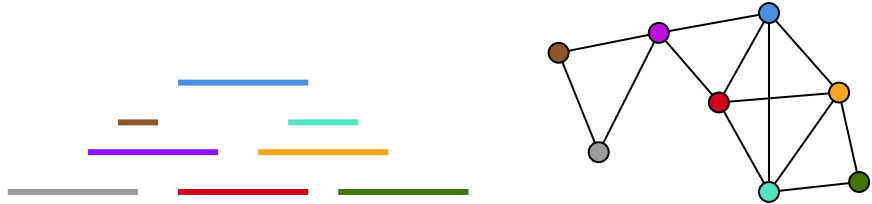
\begin{figure}[h]
	\centering
	\tikzset{every picture/.style={line width=0.75pt}} 

\begin{tikzpicture}[x=0.75pt,y=0.75pt,yscale=-1,xscale=1]

\draw    (390,75) -- (360,135) ;
\draw    (445,155) -- (420,110) ;
\draw    (420,110) -- (390,75) ;
\draw    (490,150) -- (480,105) ;
\draw    (340,85) -- (390,75) ;
\draw    (445,65) -- (420,110) ;
\draw    (445,65) -- (445,155) ;
\draw    (445,65) -- (390,75) ;
\draw    (480,105) -- (445,155) ;
\draw    (445,65) -- (480,105) ;
\draw    (360,135) -- (340,85) ;
\draw    (445,155) -- (490,150) ;
\draw    (480,105) -- (420,110) ;
\draw  [fill={rgb, 255:red, 65; green, 117; blue, 5 }  ,fill opacity=1 ] (485,150) .. controls (485,147.24) and (487.24,145) .. (490,145) .. controls (492.76,145) and (495,147.24) .. (495,150) .. controls (495,152.76) and (492.76,155) .. (490,155) .. controls (487.24,155) and (485,152.76) .. (485,150) -- cycle ;
\draw  [fill={rgb, 255:red, 245; green, 166; blue, 35 }  ,fill opacity=1 ] (475,105) .. controls (475,102.24) and (477.24,100) .. (480,100) .. controls (482.76,100) and (485,102.24) .. (485,105) .. controls (485,107.76) and (482.76,110) .. (480,110) .. controls (477.24,110) and (475,107.76) .. (475,105) -- cycle ;
\draw  [fill={rgb, 255:red, 74; green, 144; blue, 226 }  ,fill opacity=1 ] (440,65) .. controls (440,62.24) and (442.24,60) .. (445,60) .. controls (447.76,60) and (450,62.24) .. (450,65) .. controls (450,67.76) and (447.76,70) .. (445,70) .. controls (442.24,70) and (440,67.76) .. (440,65) -- cycle ;
\draw  [fill={rgb, 255:red, 80; green, 227; blue, 194 }  ,fill opacity=1 ] (440,155) .. controls (440,152.24) and (442.24,150) .. (445,150) .. controls (447.76,150) and (450,152.24) .. (450,155) .. controls (450,157.76) and (447.76,160) .. (445,160) .. controls (442.24,160) and (440,157.76) .. (440,155) -- cycle ;
\draw  [fill={rgb, 255:red, 189; green, 16; blue, 224 }  ,fill opacity=1 ] (385,75) .. controls (385,72.24) and (387.24,70) .. (390,70) .. controls (392.76,70) and (395,72.24) .. (395,75) .. controls (395,77.76) and (392.76,80) .. (390,80) .. controls (387.24,80) and (385,77.76) .. (385,75) -- cycle ;
\draw  [fill={rgb, 255:red, 208; green, 2; blue, 27 }  ,fill opacity=1 ] (415,110) .. controls (415,107.24) and (417.24,105) .. (420,105) .. controls (422.76,105) and (425,107.24) .. (425,110) .. controls (425,112.76) and (422.76,115) .. (420,115) .. controls (417.24,115) and (415,112.76) .. (415,110) -- cycle ;
\draw  [fill={rgb, 255:red, 155; green, 155; blue, 155 }  ,fill opacity=1 ] (355,135) .. controls (355,132.24) and (357.24,130) .. (360,130) .. controls (362.76,130) and (365,132.24) .. (365,135) .. controls (365,137.76) and (362.76,140) .. (360,140) .. controls (357.24,140) and (355,137.76) .. (355,135) -- cycle ;
\draw  [fill={rgb, 255:red, 139; green, 87; blue, 42 }  ,fill opacity=1 ] (335,85) .. controls (335,82.24) and (337.24,80) .. (340,80) .. controls (342.76,80) and (345,82.24) .. (345,85) .. controls (345,87.76) and (342.76,90) .. (340,90) .. controls (337.24,90) and (335,87.76) .. (335,85) -- cycle ;
\draw [color={rgb, 255:red, 155; green, 155; blue, 155 }  ,draw opacity=1 ][fill={rgb, 255:red, 248; green, 231; blue, 28 }  ,fill opacity=1 ][line width=2.25]    (65,155) -- (130,155) ;
\draw [color={rgb, 255:red, 208; green, 2; blue, 27 }  ,draw opacity=1 ][line width=2.25]    (150,155) -- (215,155) ;
\draw [color={rgb, 255:red, 65; green, 117; blue, 5 }  ,draw opacity=1 ][line width=2.25]    (230,155) -- (295,155) ;
\draw [color={rgb, 255:red, 144; green, 19; blue, 254 }  ,draw opacity=1 ][line width=2.25]    (105,135) -- (170,135) ;
\draw [color={rgb, 255:red, 245; green, 166; blue, 35 }  ,draw opacity=1 ][line width=2.25]    (190,135) -- (255,135) ;
\draw [color={rgb, 255:red, 80; green, 227; blue, 194 }  ,draw opacity=1 ][line width=2.25]    (205,120) -- (240,120) ;
\draw [color={rgb, 255:red, 139; green, 87; blue, 42 }  ,draw opacity=1 ][line width=2.25]    (120,120) -- (140,120) ;
\draw [color={rgb, 255:red, 74; green, 144; blue, 226 }  ,draw opacity=1 ][line width=2.25]    (150,100) -- (215,100) ;

\end{tikzpicture}
	\caption{An example of an interval graph together with a representation as the intersection of intervals.}
	\label{fig:intervalEx}
\end{figure}
	
\begin{figure}[h]
	\centering 
	\tikzset{every picture/.style={line width=0.75pt}} 

\begin{tikzpicture}[x=0.75pt,y=0.75pt,yscale=-1,xscale=1]

\draw [color={rgb, 255:red, 208; green, 2; blue, 27 }  ,draw opacity=1 ][line width=1.5]    (205,135) -- (245,60) ;
\draw [color={rgb, 255:red, 245; green, 166; blue, 35 }  ,draw opacity=1 ][line width=1.5]    (305,135) -- (225,60) ;
\draw [color={rgb, 255:red, 248; green, 231; blue, 28 }  ,draw opacity=1 ][line width=1.5]    (245,135) -- (205,60) ;
\draw [color={rgb, 255:red, 189; green, 16; blue, 224 }  ,draw opacity=1 ][line width=1.5]    (225,135) -- (305,60) ;
\draw [color={rgb, 255:red, 74; green, 144; blue, 226 }  ,draw opacity=1 ][line width=1.5]    (265,135) -- (325,60) ;
\draw [color={rgb, 255:red, 155; green, 155; blue, 155 }  ,draw opacity=1 ][line width=1.5]    (325,135) -- (285,60) ;
\draw [color={rgb, 255:red, 126; green, 211; blue, 33 }  ,draw opacity=1 ][line width=1.5]    (285,135) -- (265,60) ;
\draw    (380,75) -- (385,120) -- (410,110) -- (415,70) -- (475,65) -- (455,90) ;
\draw    (475,65) -- (475,120) ;
\draw    (415,70) -- (380,75) ;
\draw    (475,120) -- (410,110) ;
\draw    (455,90) -- (410,110) ;
\draw    (415,70) -- (455,90) ;
\draw  [fill={rgb, 255:red, 126; green, 211; blue, 33 }  ,fill opacity=1 ] (450,90) .. controls (450,87.24) and (452.24,85) .. (455,85) .. controls (457.76,85) and (460,87.24) .. (460,90) .. controls (460,92.76) and (457.76,95) .. (455,95) .. controls (452.24,95) and (450,92.76) .. (450,90) -- cycle ;
\draw  [fill={rgb, 255:red, 248; green, 231; blue, 28 }  ,fill opacity=1 ] (380,120) .. controls (380,117.24) and (382.24,115) .. (385,115) .. controls (387.76,115) and (390,117.24) .. (390,120) .. controls (390,122.76) and (387.76,125) .. (385,125) .. controls (382.24,125) and (380,122.76) .. (380,120) -- cycle ;
\draw  [fill={rgb, 255:red, 208; green, 2; blue, 27 }  ,fill opacity=1 ] (375,75) .. controls (375,72.24) and (377.24,70) .. (380,70) .. controls (382.76,70) and (385,72.24) .. (385,75) .. controls (385,77.76) and (382.76,80) .. (380,80) .. controls (377.24,80) and (375,77.76) .. (375,75) -- cycle ;
\draw  [fill={rgb, 255:red, 189; green, 16; blue, 224 }  ,fill opacity=1 ] (405,110) .. controls (405,107.24) and (407.24,105) .. (410,105) .. controls (412.76,105) and (415,107.24) .. (415,110) .. controls (415,112.76) and (412.76,115) .. (410,115) .. controls (407.24,115) and (405,112.76) .. (405,110) -- cycle ;
\draw  [fill={rgb, 255:red, 245; green, 166; blue, 35 }  ,fill opacity=1 ] (410,70) .. controls (410,67.24) and (412.24,65) .. (415,65) .. controls (417.76,65) and (420,67.24) .. (420,70) .. controls (420,72.76) and (417.76,75) .. (415,75) .. controls (412.24,75) and (410,72.76) .. (410,70) -- cycle ;
\draw  [fill={rgb, 255:red, 155; green, 155; blue, 155 }  ,fill opacity=1 ] (470,120) .. controls (470,117.24) and (472.24,115) .. (475,115) .. controls (477.76,115) and (480,117.24) .. (480,120) .. controls (480,122.76) and (477.76,125) .. (475,125) .. controls (472.24,125) and (470,122.76) .. (470,120) -- cycle ;
\draw  [fill={rgb, 255:red, 74; green, 144; blue, 226 }  ,fill opacity=1 ] (470,65) .. controls (470,62.24) and (472.24,60) .. (475,60) .. controls (477.76,60) and (480,62.24) .. (480,65) .. controls (480,67.76) and (477.76,70) .. (475,70) .. controls (472.24,70) and (470,67.76) .. (470,65) -- cycle ;
\draw [line width=0.75]    (205,135) -- (325,135) ;
\draw [line width=0.75]    (205,60) -- (325,60) ;
\draw [line width=0.75]    (205,55) -- (205,65) ;
\draw [line width=0.75]    (225,55) -- (225,65) ;
\draw [line width=0.75]    (245,55) -- (245,65) ;
\draw [line width=0.75]    (265,55) -- (265,65) ;
\draw [line width=0.75]    (285,55) -- (285,65) ;
\draw [line width=0.75]    (305,55) -- (305,65) ;
\draw [line width=0.75]    (325,55) -- (325,65) ;
\draw [line width=0.75]    (205,130) -- (205,140) ;
\draw [line width=0.75]    (225,130) -- (225,140) ;
\draw [line width=0.75]    (245,130) -- (245,140) ;
\draw [line width=0.75]    (265,130) -- (265,140) ;
\draw [line width=0.75]    (285,130) -- (285,140) ;
\draw [line width=0.75]    (305,130) -- (305,140) ;
\draw [line width=0.75]    (325,130) -- (325,140) ;

\end{tikzpicture}
	\caption{An example of a permutation graph with its corresponding intersection model.}
	\label{fig:Expermutation}
\end{figure}
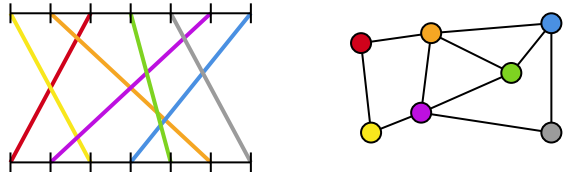

Although the class of interval graphs is quite restrictive, there are a number of practical applications and specialized algorithms for interval graphs \cite{golumbic2004algorithmic,halldorsson2020improved,konrad2019distributed}. 
Moreover, for several applications, the subclass of unit interval graphs (the situation where all the intervals have the same length)
turns out to be extremely useful as well~\cite{beeri1983desirability,kaplan1996pathwidth}.

A  natural generalization of interval graphs are {\emph{circular arc graphs}}, where the segments, instead of lying  
on a line, lie on a circle. More precisely, a circular arc graph is the intersection graph of arcs of a circle (see Figure~\ref{fig:circleArcEx}).
Although circular arc graphs look similar to interval graphs, several combinatorial problems behave very differently on these two classes of graphs.  
For example, the coloring problem is NP-complete for circular-arc graphs while it can be solved in linear time on interval graphs \cite{garey1980complexity}.
Recognizing circular-arc graphs can also be done in linear time \cite{kaplan2006simpler,mcconnell2003linear}.

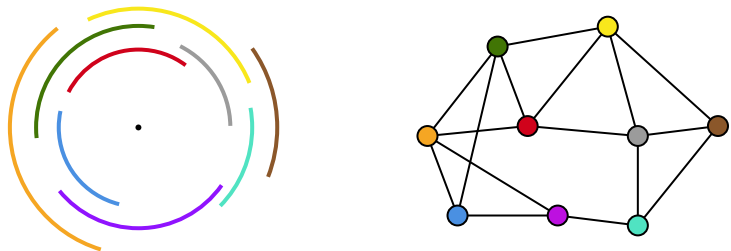
\begin{figure}[!ht]
	\centering
	\tikzset{every picture/.style={line width=0.75pt}} 

\begin{tikzpicture}[x=0.75pt,y=0.75pt,yscale=-1,xscale=1]

\draw  [draw opacity=0][fill={rgb, 255:red, 0; green, 0; blue, 0 }  ,fill opacity=1 ][line width=0.75]  (204.23,119.29) .. controls (204.23,118.47) and (204.89,117.81) .. (205.71,117.81) .. controls (206.53,117.81) and (207.19,118.47) .. (207.19,119.29) .. controls (207.19,120.11) and (206.53,120.77) .. (205.71,120.77) .. controls (204.89,120.77) and (204.23,120.11) .. (204.23,119.29) -- cycle ;
\draw  [draw opacity=0][line width=1.5]  (186.94,180.67) .. controls (160.66,172.65) and (141.54,148.2) .. (141.54,119.29) .. controls (141.54,99.16) and (150.82,81.19) .. (165.33,69.42) -- (205.71,119.29) -- cycle ; \draw  [color={rgb, 255:red, 245; green, 166; blue, 35 }  ,draw opacity=1 ][line width=1.5]  (186.94,180.67) .. controls (160.66,172.65) and (141.54,148.2) .. (141.54,119.29) .. controls (141.54,99.16) and (150.82,81.19) .. (165.33,69.42) ;  
\draw  [draw opacity=0][line width=1.5]  (154.9,124.63) .. controls (154.72,122.88) and (154.63,121.1) .. (154.63,119.29) .. controls (154.63,91.08) and (177.5,68.21) .. (205.71,68.21) .. controls (208.43,68.21) and (211.1,68.42) .. (213.7,68.83) -- (205.71,119.29) -- cycle ; \draw  [color={rgb, 255:red, 65; green, 117; blue, 5 }  ,draw opacity=1 ][line width=1.5]  (154.9,124.63) .. controls (154.72,122.88) and (154.63,121.1) .. (154.63,119.29) .. controls (154.63,91.08) and (177.5,68.21) .. (205.71,68.21) .. controls (208.43,68.21) and (211.1,68.42) .. (213.7,68.83) ;  
\draw  [draw opacity=0][line width=1.5]  (170.8,101.51) .. controls (177.28,88.82) and (190.48,80.12) .. (205.71,80.12) .. controls (214.56,80.12) and (222.72,83.06) .. (229.28,88.01) -- (205.71,119.29) -- cycle ; \draw  [color={rgb, 255:red, 208; green, 2; blue, 27 }  ,draw opacity=1 ][line width=1.5]  (170.8,101.51) .. controls (177.28,88.82) and (190.48,80.12) .. (205.71,80.12) .. controls (214.56,80.12) and (222.72,83.06) .. (229.28,88.01) ;  
\draw  [draw opacity=0][line width=1.5]  (196.08,157.89) .. controls (178.77,153.59) and (165.94,137.94) .. (165.94,119.29) .. controls (165.94,116.46) and (166.23,113.69) .. (166.8,111.02) -- (205.71,119.29) -- cycle ; \draw  [color={rgb, 255:red, 74; green, 144; blue, 226 }  ,draw opacity=1 ][line width=1.5]  (196.08,157.89) .. controls (178.77,153.59) and (165.94,137.94) .. (165.94,119.29) .. controls (165.94,116.46) and (166.23,113.69) .. (166.8,111.02) ;  
\draw  [draw opacity=0][line width=1.5]  (247.24,148.37) .. controls (238.07,161.44) and (222.89,169.98) .. (205.71,169.98) .. controls (189.8,169.98) and (175.61,162.66) .. (166.31,151.19) -- (205.71,119.29) -- cycle ; \draw  [color={rgb, 255:red, 144; green, 19; blue, 254 }  ,draw opacity=1 ][line width=1.5]  (247.24,148.37) .. controls (238.07,161.44) and (222.89,169.98) .. (205.71,169.98) .. controls (189.8,169.98) and (175.61,162.66) .. (166.31,151.19) ;  
\draw  [draw opacity=0][line width=1.5]  (180.44,65.11) .. controls (188.12,61.52) and (196.68,59.52) .. (205.71,59.52) .. controls (230.8,59.52) and (252.28,74.98) .. (261.14,96.9) -- (205.71,119.29) -- cycle ; \draw  [color={rgb, 255:red, 248; green, 231; blue, 28 }  ,draw opacity=1 ][line width=1.5]  (180.44,65.11) .. controls (188.12,61.52) and (196.68,59.52) .. (205.71,59.52) .. controls (230.8,59.52) and (252.28,74.98) .. (261.14,96.9) ;  
\draw  [draw opacity=0][line width=1.5]  (261.65,109.43) .. controls (262.21,112.63) and (262.5,115.93) .. (262.5,119.29) .. controls (262.5,134.62) and (256.43,148.53) .. (246.56,158.74) -- (205.71,119.29) -- cycle ; \draw  [color={rgb, 255:red, 80; green, 227; blue, 194 }  ,draw opacity=1 ][line width=1.5]  (261.65,109.43) .. controls (262.21,112.63) and (262.5,115.93) .. (262.5,119.29) .. controls (262.5,134.62) and (256.43,148.53) .. (246.56,158.74) ;  
\draw  [draw opacity=0][line width=1.5]  (262.47,79.54) .. controls (270.37,90.8) and (275,104.5) .. (275,119.29) .. controls (275,128.05) and (273.38,136.42) .. (270.42,144.13) -- (205.71,119.29) -- cycle ; \draw  [color={rgb, 255:red, 139; green, 87; blue, 42 }  ,draw opacity=1 ][line width=1.5]  (262.47,79.54) .. controls (270.37,90.8) and (275,104.5) .. (275,119.29) .. controls (275,128.05) and (273.38,136.42) .. (270.42,144.13) ;  
\draw  [draw opacity=0][line width=1.5]  (226.54,78.4) .. controls (241.19,85.88) and (251.29,100.99) .. (251.59,118.49) -- (205.71,119.29) -- cycle ; \draw  [color={rgb, 255:red, 155; green, 155; blue, 155 }  ,draw opacity=1 ][line width=1.5]  (226.54,78.4) .. controls (241.19,85.88) and (251.29,100.99) .. (251.59,118.49) ;  
\draw    (350,123.58) -- (365,163.58) ;
\draw    (415,163.58) -- (350,123.58) ;
\draw    (400,118.58) -- (350,123.58) ;
\draw    (385,78.58) -- (400,118.58) ;
\draw    (385,78.58) -- (365,163.58) ;
\draw    (455,168.58) -- (415,163.58) ;
\draw    (440,68.58) -- (455,123.58) ;
\draw    (440,68.58) -- (400,118.58) ;
\draw    (440,68.58) -- (385,78.58) ;
\draw    (440,68.58) -- (495,118.58) ;
\draw    (415,163.58) -- (365,163.58) ;
\draw    (495,118.58) -- (455,168.58) ;
\draw    (495,118.58) -- (455,123.58) ;
\draw    (455,123.58) -- (455,168.58) ;
\draw    (400,118.58) -- (455,123.58) ;
\draw    (385,78.58) -- (350,123.58) ;
\draw  [fill={rgb, 255:red, 65; green, 117; blue, 5 }  ,fill opacity=1 ] (380,78.58) .. controls (380,75.82) and (382.24,73.58) .. (385,73.58) .. controls (387.76,73.58) and (390,75.82) .. (390,78.58) .. controls (390,81.34) and (387.76,83.58) .. (385,83.58) .. controls (382.24,83.58) and (380,81.34) .. (380,78.58) -- cycle ;
\draw  [fill={rgb, 255:red, 245; green, 166; blue, 35 }  ,fill opacity=1 ] (345,123.58) .. controls (345,120.82) and (347.24,118.58) .. (350,118.58) .. controls (352.76,118.58) and (355,120.82) .. (355,123.58) .. controls (355,126.34) and (352.76,128.58) .. (350,128.58) .. controls (347.24,128.58) and (345,126.34) .. (345,123.58) -- cycle ;
\draw  [fill={rgb, 255:red, 74; green, 144; blue, 226 }  ,fill opacity=1 ] (360,163.58) .. controls (360,160.82) and (362.24,158.58) .. (365,158.58) .. controls (367.76,158.58) and (370,160.82) .. (370,163.58) .. controls (370,166.34) and (367.76,168.58) .. (365,168.58) .. controls (362.24,168.58) and (360,166.34) .. (360,163.58) -- cycle ;
\draw  [fill={rgb, 255:red, 80; green, 227; blue, 194 }  ,fill opacity=1 ] (450,168.58) .. controls (450,165.82) and (452.24,163.58) .. (455,163.58) .. controls (457.76,163.58) and (460,165.82) .. (460,168.58) .. controls (460,171.34) and (457.76,173.58) .. (455,173.58) .. controls (452.24,173.58) and (450,171.34) .. (450,168.58) -- cycle ;
\draw  [fill={rgb, 255:red, 189; green, 16; blue, 224 }  ,fill opacity=1 ] (410,163.58) .. controls (410,160.82) and (412.24,158.58) .. (415,158.58) .. controls (417.76,158.58) and (420,160.82) .. (420,163.58) .. controls (420,166.34) and (417.76,168.58) .. (415,168.58) .. controls (412.24,168.58) and (410,166.34) .. (410,163.58) -- cycle ;
\draw  [fill={rgb, 255:red, 208; green, 2; blue, 27 }  ,fill opacity=1 ] (395,118.58) .. controls (395,115.82) and (397.24,113.58) .. (400,113.58) .. controls (402.76,113.58) and (405,115.82) .. (405,118.58) .. controls (405,121.34) and (402.76,123.58) .. (400,123.58) .. controls (397.24,123.58) and (395,121.34) .. (395,118.58) -- cycle ;
\draw  [fill={rgb, 255:red, 248; green, 231; blue, 28 }  ,fill opacity=1 ] (435,68.58) .. controls (435,65.82) and (437.24,63.58) .. (440,63.58) .. controls (442.76,63.58) and (445,65.82) .. (445,68.58) .. controls (445,71.34) and (442.76,73.58) .. (440,73.58) .. controls (437.24,73.58) and (435,71.34) .. (435,68.58) -- cycle ;
\draw  [fill={rgb, 255:red, 155; green, 155; blue, 155 }  ,fill opacity=1 ] (450,123.58) .. controls (450,120.82) and (452.24,118.58) .. (455,118.58) .. controls (457.76,118.58) and (460,120.82) .. (460,123.58) .. controls (460,126.34) and (457.76,128.58) .. (455,128.58) .. controls (452.24,128.58) and (450,126.34) .. (450,123.58) -- cycle ;
\draw  [fill={rgb, 255:red, 139; green, 87; blue, 42 }  ,fill opacity=1 ] (490,118.58) .. controls (490,115.82) and (492.24,113.58) .. (495,113.58) .. controls (497.76,113.58) and (500,115.82) .. (500,118.58) .. controls (500,121.34) and (497.76,123.58) .. (495,123.58) .. controls (492.24,123.58) and (490,121.34) .. (490,118.58) -- cycle ;
\end{tikzpicture}
	\caption{An example of a circular arc graph: the left side shows a representation with overlapping arcs in the circle, while the right side shows its associated graph realization.}
	\label{fig:circleArcEx}
\end{figure}

Another natural, well-known generalization of interval graphs is {\emph{chordal graphs}}. These graphs are intersections of subtrees of a tree. More precisely, $G$ is chordal if and only if there exists a tree $T$ such that every node of $G$ can be associated with a subtree of $T$  in such a way that two nodes of $G$ are adjacent if their corresponding subtrees intersect.  Chordal graphs are among the 
most-studied graph classes \cite{blair1993introduction, golumbic2004algorithmic} and, in fact, they have appeared in the literature with different names such as rigid-circuit graphs, triangulated graphs, 
perfect elimination graphs, decomposable graphs, acyclic graphs, etc. 
Chordal graphs can be recognized in linear time \cite{rose1976algorithmic} and they have many applications, for instance in phylogeny tree reconstruction,  a fundamental problem in 
computational biology \cite{bodlaender1992two,kennedy2006strictly,lin2000phylogenetic}. 
The name chordal comes from the fact that a graph is chordal if and only if every cycle of length at least 4 has a chord. It is interesting to point out that, in the framework of distributed computing,  the authors in \cite{bousquet2021distributed} exhibit distributed algorithms for recoloring interval and chordal graphs. 

\begin{figure}[!ht]
	\centering
	\includegraphics[width=0.6\linewidth]{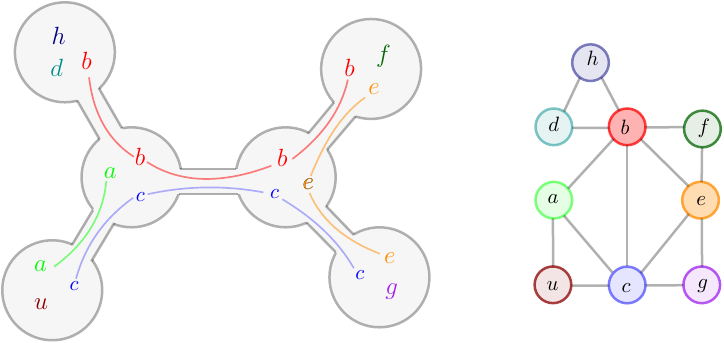}
	\caption{An example of a chordal graph with its corresponding intersection model.}
	\label{fig:trim}
\end{figure}

In addition, the class of {\emph{trapezoid graphs}}  is a  generalization of both interval graphs and permutation graphs.
A trapezoid graph  is defined as the intersection graph of trapezoids between two horizontal parallel lines with two vertices in each line
(see Figure~\ref{fig:Extrapezoid}). 
Ma and Spinrad \cite{ma19942} showed that trapezoid graphs can be recognized in $\cO(n^2)$ time.
Trapezoid graphs were applied in various contexts such as  VLSI design \cite{dagan1988trapezoid} and
bioinformatics \cite{abouelhoda2005chaining}. 

\begin{figure}[h]
\centering 

\tikzset{every picture/.style={line width=0.75pt}} 

\begin{tikzpicture}[x=0.75pt,y=0.75pt,yscale=-1,xscale=1]

\draw  [color={rgb, 255:red, 74; green, 144; blue, 226 }  ,draw opacity=1 ][fill={rgb, 255:red, 74; green, 144; blue, 226 }  ,fill opacity=0.2 ] (350,60) -- (390,165) -- (310,165) -- (230,60) -- cycle ;
\draw  [color={rgb, 255:red, 189; green, 16; blue, 224 }  ,draw opacity=1 ][fill={rgb, 255:red, 189; green, 16; blue, 224 }  ,fill opacity=0.2 ] (370,60) -- (370,165) -- (270,165) -- (270,60) -- cycle ;
\draw  [color={rgb, 255:red, 248; green, 231; blue, 28 }  ,draw opacity=1 ][fill={rgb, 255:red, 248; green, 231; blue, 28 }  ,fill opacity=0.2 ] (250,60) -- (250,165) -- (170,165) -- (150,60) -- cycle ;
\draw  [color={rgb, 255:red, 245; green, 166; blue, 35 }  ,draw opacity=1 ][fill={rgb, 255:red, 245; green, 166; blue, 35 }  ,fill opacity=0.2 ] (410,60) -- (410,165) -- (330,165) -- (330,60) -- cycle ;
\draw  [color={rgb, 255:red, 128; green, 128; blue, 128 }  ,draw opacity=1 ][fill={rgb, 255:red, 128; green, 128; blue, 128 }  ,fill opacity=0.2 ] (390,60) -- (350,165) -- (230,165) -- (310,60) -- cycle ;
\draw    (480,57) -- (620,57) -- (550,167) -- cycle ;
\draw    (535,119) -- (550,167) -- (565,119) ;
\draw    (550,84) -- (480,57) -- (535,119) -- (550,107) ;
\draw    (550,84) -- (620,57) -- (565,119) -- (550,107) -- cycle ;
\draw    (535,119) -- (565,119) -- (550,84) -- cycle ;
\draw  [fill={rgb, 255:red, 126; green, 211; blue, 33 }  ,fill opacity=1 ] (475,57) .. controls (475,54.24) and (477.24,52) .. (480,52) .. controls (482.76,52) and (485,54.24) .. (485,57) .. controls (485,59.76) and (482.76,62) .. (480,62) .. controls (477.24,62) and (475,59.76) .. (475,57) -- cycle ;
\draw  [fill={rgb, 255:red, 248; green, 231; blue, 28 }  ,fill opacity=1 ] (615,57) .. controls (615,54.24) and (617.24,52) .. (620,52) .. controls (622.76,52) and (625,54.24) .. (625,57) .. controls (625,59.76) and (622.76,62) .. (620,62) .. controls (617.24,62) and (615,59.76) .. (615,57) -- cycle ;
\draw  [fill={rgb, 255:red, 208; green, 2; blue, 27 }  ,fill opacity=1 ] (545,167) .. controls (545,164.24) and (547.24,162) .. (550,162) .. controls (552.76,162) and (555,164.24) .. (555,167) .. controls (555,169.76) and (552.76,172) .. (550,172) .. controls (547.24,172) and (545,169.76) .. (545,167) -- cycle ;
\draw  [fill={rgb, 255:red, 189; green, 16; blue, 224 }  ,fill opacity=1 ] (530,119) .. controls (530,116.24) and (532.24,114) .. (535,114) .. controls (537.76,114) and (540,116.24) .. (540,119) .. controls (540,121.76) and (537.76,124) .. (535,124) .. controls (532.24,124) and (530,121.76) .. (530,119) -- cycle ;
\draw  [fill={rgb, 255:red, 245; green, 166; blue, 35 }  ,fill opacity=1 ] (545,107) .. controls (545,104.24) and (547.24,102) .. (550,102) .. controls (552.76,102) and (555,104.24) .. (555,107) .. controls (555,109.76) and (552.76,112) .. (550,112) .. controls (547.24,112) and (545,109.76) .. (545,107) -- cycle ;
\draw  [fill={rgb, 255:red, 155; green, 155; blue, 155 }  ,fill opacity=1 ] (560,119) .. controls (560,116.24) and (562.24,114) .. (565,114) .. controls (567.76,114) and (570,116.24) .. (570,119) .. controls (570,121.76) and (567.76,124) .. (565,124) .. controls (562.24,124) and (560,121.76) .. (560,119) -- cycle ;
\draw  [fill={rgb, 255:red, 74; green, 144; blue, 226 }  ,fill opacity=1 ] (545,84) .. controls (545,81.24) and (547.24,79) .. (550,79) .. controls (552.76,79) and (555,81.24) .. (555,84) .. controls (555,86.76) and (552.76,89) .. (550,89) .. controls (547.24,89) and (545,86.76) .. (545,84) -- cycle ;
\draw  [color={rgb, 255:red, 208; green, 2; blue, 27 }  ,draw opacity=1 ][fill={rgb, 255:red, 208; green, 2; blue, 27 }  ,fill opacity=0.2 ] (210,60) -- (290,165) -- (190,165) -- (170,60) -- cycle ;
\draw  [color={rgb, 255:red, 65; green, 117; blue, 5 }  ,draw opacity=1 ][fill={rgb, 255:red, 126; green, 211; blue, 33 }  ,fill opacity=0.2 ] (290,60) -- (210,165) -- (150,165) -- (190,60) -- cycle ;
\draw [line width=0.75]    (150,60) -- (401.83,60) -- (410,60) ;
\draw [line width=0.75]    (150,165) -- (410,165) ;
\draw [line width=0.75]    (290,160) -- (290,170) ;
\draw [line width=0.75]    (310,160) -- (310,170) ;
\draw [line width=0.75]    (330,160) -- (330,170) ;
\draw [line width=0.75]    (350,160) -- (350,170) ;
\draw [line width=0.75]    (370,160) -- (370,170) ;
\draw [line width=0.75]    (390,160) -- (390,170) ;
\draw [line width=0.75]    (410,160) -- (410,170) ;
\draw [line width=0.75]    (150,160) -- (150,170) ;
\draw [line width=0.75]    (170,160) -- (170,170) ;
\draw [line width=0.75]    (190,160) -- (190,170) ;
\draw [line width=0.75]    (210,160) -- (210,170) ;
\draw [line width=0.75]    (230,160) -- (230,170) ;
\draw [line width=0.75]    (250,160) -- (250,170) ;
\draw [line width=0.75]    (270,160) -- (270,170) ;
\draw [line width=0.75]    (290,55) -- (290,65) ;
\draw [line width=0.75]    (310,55) -- (310,65) ;
\draw [line width=0.75]    (330,55) -- (330,65) ;
\draw [line width=0.75]    (350,55) -- (350,65) ;
\draw [line width=0.75]    (370,55) -- (370,65) ;
\draw [line width=0.75]    (390,55) -- (390,65) ;
\draw [line width=0.75]    (410,55) -- (410,65) ;
\draw [line width=0.75]    (150,55) -- (150,65) ;
\draw [line width=0.75]    (170,55) -- (170,65) ;
\draw [line width=0.75]    (190,55) -- (190,65) ;
\draw [line width=0.75]    (210,55) -- (210,65) ;
\draw [line width=0.75]    (230,55) -- (230,65) ;
\draw [line width=0.75]    (250,55) -- (250,65) ;
\draw [line width=0.75]    (270,55) -- (270,65) ;

\end{tikzpicture}

\caption{An example of a trapezoid graph with its corresponding intersection model.}
\label{fig:Extrapezoid}
\end{figure}
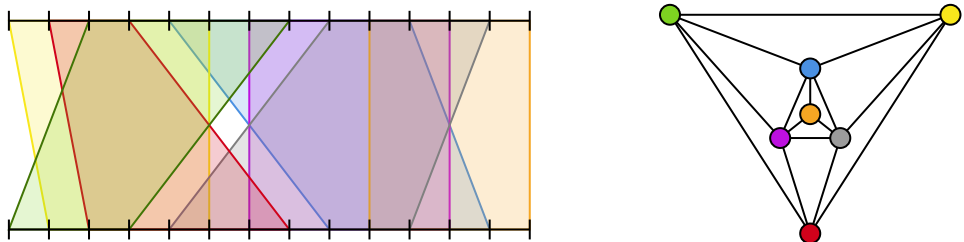

\subsection{PLSs for Geometric Graph Classes}

A naive approach to defining a PLS for a geometric graph class is to assign each vertex the corresponding geometric object it represents. During the verification phase, the vertices could check with their neighbors to ensure that the objects they represent intersect. However, this naive approach is not generally effective in defining compact certificates due to two difficulties.

First, we would need to encode the geometric objects using a logarithmic number of bits. While this may be possible for certain geometric graph classes, such as interval graphs, it is not clear if it holds true in general. For example, for chordal graphs, we do not know how to encode subtrees of a given tree using a logarithmic number of bits with respect to its size.

The second difficulty is that even if we could efficiently encode the objects, the soundness condition of a PLS requires that every graph not belonging to the given geometric graph class has to be rejected by the certification process. Therefore, to satisfy the soundness requirement, the vertices would also need to check that all non-adjacent vertices are assigned non-intersecting objects. This would require a vertex to check conditions with other non-adjacent vertices that could be far away in the graph. Notice that geometric graph classes, including the ones discussed in this article, can have arbitrarily large diameters.

Therefore, it is necessary to develop more sophisticated ideas in order to overcome the difficulties inherent in the naive approach.

\subsection{Our Results and Techniques}

In the present work we show compact PLSs (i.e. with logarithmic-sized certificates) for the recognition of all the aforementioned geometric intersection
graph classes, namely interval and chordal graphs (Section~\ref{sec:interval}), circular arc graphs (Section~\ref{sec:circular}) and, finally,  trapezoid and permutation graphs (Section~\ref{sec:trapezoid}). For all these classes we also provide, in Section~\ref{sec:lower}, tight logarithmic lower bounds on the size of the certificates.\\

In our results, we employ different sets of techniques that leverage the structural properties of the considered graph classes. We will now briefly explain our constructions.\\

{\it Chordal and Interval graphs.} As we explained above, chordal graphs are intersections of subtrees of a tree. They are also defined as the graph where every induced cycle has length at most \(3\) (i.e. every non-induced cycle has a \emph{chord}). Interestingly, chordal graphs can be characterized by the existence of a specific tree-decomposition, called \emph{clique-tree}. This tree-decomposition shares the same properties that define the treewidth, with the exception that each bag (node of the decomposition) forms a maximal clique of the graph (see Section~\ref{sec:interval} for more details on these definitions).  The certification of the clique-tree shares some ideas with the ones used in~\cite{FraigniaudMRT22} to certify graphs of bounded tree-with. Observe, however, that the maximal cliques of a chordal graph are unbounded, as a chordal graph may have unbounded tree-width. Therefore, new ideas had to be developed. We take advantage of the properties of the clique-trees, in particular high connectivity within the bags, to obtain a PLS with certificates of size \(\cO(\log n)\) for the verification of chordal graphs. 

The PLS for certification of interval graphs follows as a direct application of the PLS for chordal graphs. Indeed, an interval graph is a particular type of chordal graph, where the clique-tree is restricted to be a path. Then, the certification of interval graphs uses the certification of chordal graphs, while at the same time verifying that the given decomposition is indeed a path. \\

\textit{Circular Arc graphs.} As described in \Cref{fig:circleArcEx}, a circular-arc graph is represented by a set of arcs in a circle such that two nodes are neighbors if and only if their arcs have no empty intersection. In order to recognize this class, we first tackle the problem of recognizing the subclass of \textit{proper circular-arc} graphs, which are graphs that admit a circular-arc representation such that no arc is contained in other. Using a property of the adjacency matrix of graphs in this subclass given by ~\cite{tucker1970}, we develop an algorithm to recognize this property. Then, we expand the property of ~\cite{tucker1970} to the whole class of circular arc graphs and we proceed to verify this property distributively. In this part, the main idea is to develop an algorithm that allows the nodes to verify a global property of their adjacency matrix, which is an extension of a known characterization for proper circular-arc graphs.\\

{\it Trapezoid and Permutation graphs.} Recall that a graph is a trapezoid graph if each node can be assigned to a trapezoid inscribed in two parallel lines, with two vertices in each line, such that two nodes are neighbors if and only if their corresponding trapezoids have no empty intersection, as shown in \Cref{fig:Extrapezoid}. As both parallel lines contain $2n$ vertices, we can enumerate these endpoints on each line from $1$ to $2n$, so a trapezoid can be characterized as a tuple $(t_1(v),t_2(v),b_1(v),b_2(v))\in [2n]^4$, corresponding to the enumeration of each vertex. If the collection $\{(t_1(v),t_2(v),b_1(v),b_2(v))\}_{v\in V}$ satisfy that two nodes $u,v$ are neighbors if and only if their corresponding trapezoids intersect, we say it is a \textit{ proper trapezoid model}, and if only satisfy that all neighbors have non-empty trapezoid (one part of the equivalence), we call it a \textit{semi proper trapezoid model}. Then, verifying that a given model $\{(t_1(v),t_2(v),b_1(v),b_2(v))\}_{v\in V}$ given by the prover is a semi-proper model it is straightforward to do distributively: each node shares its trapezoid model with its neighbors and check they intersect. Then, in order to prove that a semi-trapezoid model is a trapezoid model, we need to verify that all non-adjacent nodes have empty trapezoid intersections. As we cannot do this directly, because we don't have direct communication between no adjacent nodes, we prove that a semi-trapezoid model is a proper trapezoid model if it satisfies two conditions (\Cref{lem:nontrapcarac2}), which are easier to verify distributively, because are dependent on the positions between their vertices in each line, and a local calculation that can be computed by each node.

Then, the result implies a PLS to recognize permutation graphs because we prove that a permutation model, i.e., a collection of lines with endpoints in two parallel lines, as in \Cref{fig:Expermutation}, such that each node is associated with a line, and two nodes are neighbors if an only if their corresponding lines intersects, it can be represented as a specific proper trapezoid model with an extra condition that can be verified locally by the nodes.

{\it Lower bounds.} To obtain tight lower bounds we use two different approaches. First, to get a lower bound for the classes of interval, circular arc and chordal graphs we adapt a construction for lower bounds in the Locally  Proof model from ~\cite{goos2016locally} to the PLS model, where the main idea is to construct a collection of graphs in each class that would be indistinguishable from a particular no-instance if we allow messages of just $o(\log n)$ bits. In order to obtain a lower bound of $\Omega(\log n)$ proof-size for the recognition of permutation and trapezoid graphs, we use a technique from ~\cite{fraigniaud2019randomized} called \textit{crossing edge}, in which we need to construct a specific graph that is part of the class, but if we interchange specific edges between some nodes, the resulting graph is no longer part of the class. Then, by a result of ~\cite{fraigniaud2019randomized}, we have the desired tight lower bound.

\subsection{Related Work}

Since the introduction of PLSs~\cite{korman2010proof}, different variants were introduced.  Some stronger forms of PLS include locally checkable proofs~\cite{goos2016locally}, where each node can send not only its certificates, but also its state,  and $t$-PLS~\cite{FeuilloleyFHPP21}, where nodes perform communication at distance~$t\geq 1$ before deciding. Authors have studied many other variants of PLSs, such as randomized PLSs \cite{fraigniaud2019randomized}, quantum PLSs~\cite{FraigniaudGNP21}, interactive protocols~\cite{CrescenziFP19,kol2018interactive,NaorPY20}, zero-knowledge distributed certification~\cite{BickKO22}, 
PLSs use global certificates in addition to the local ones~\cite{FeuilloleyH18}, etc. On the other hand, some trade-offs between the size of the certificates and the number of rounds of the verification protocol have been exhibited ~\cite{FeuilloleyFHPP21}. Also, several hierarchies of certification mechanisms have been introduced, including games between a prover and a disprover~\cite{BalliuDFO18,FeuilloleyFH21}.

PLSs have been shown to be effective for recognizing many graph classes. For example, there are compact PLSs  (i.e. with logarithmic size certificates) for the recognition of acyclic graphs \cite{KormanKP10}, planar graphs~\cite{feuilloley2020compact}, graphs with bounded genus~\cite{EsperetL22}, \(H\)-minor-free graphs (as long as \(H\) has at most four vertices)~\cite{BousquetFP21}, etc. 

In a recent breakthrough, Bousquet et al.~\cite{bousquet2021local}  proved a ``meta-theorem'',  stating that, there exists a PLS for deciding any monadic second-order logic property with $O(\log n)$-bit certificates on graphs of bounded \emph{tree-depth}. This result has been extended by Fraigniaud et al~\cite{FraigniaudMRT22} to the larger class of graphs with bounded \emph{tree-width}, using certificates on $O(\log^2 n)$ bits. This result implies in particular the existence of (nearly) compact PLS  for certifying the class of graphs with tree-width at most $k$ (for any fixed $k$). Moreover, these results have other direct implications for the design and analysis of (nearly) compact PLSs for graphs with certain structural properties. For instance, for every planar graph \(H\), there is a PLS verifying \(H\)-minor free graphs with certificates of size \(\cO(\log^2 n)\).

\section{Preliminaries}
All graphs in this work are considered simple and undirected. An $n$-node graph $G=(V,E)$ is a graph with $|V|=n$. Given a graph $G=(V,E)$, the set of neighbors of a node $v\in V$ (nodes connected to $v$ via an edge en $G$) is denoted as $N_G(v)$\footnote{When the graph is clear by context by omit the subscript}.

Given $n\in\N$, $[n]$ corresponds to the set $\{1,...,n\}$ and $S_n$ to the set of all permutations in $[n]$. For $n,m\in \N$, $n<m$, we define $[n,m]_\N =\{n,n+1,...,m-1,m\}$.

\subsection{Distributed Languages}

Let $G=(V,E)$ be a simple connected $n$-node graph, let $I\colon V\to \{0,1\}^*$ be an input function assigning labels to the nodes of $G$. where the size of all inputs is polynomially bounded on $n$. Let $\id\colon V\to\{1,...,n^c\}$ for some constant $c>0$ be a one-to-one function assigning identifiers to the nodes. A \textit{distributed language} $\mathcal{L}$ is a (Turing decidable) collection of triples $(G,\id,I)$, called \textit{network configurations}.

Sometimes the label function $I$ represents some construction over the graph, for example, it can be a single bit in $\{0,1\}$ indicating a subset of nodes, which can represent a vertex cover set, maximal independent set, etc. In our case, we are interested in a property of $G$ itself, and not in verifying some property over the labels given by $I$, so even if the formal definition defines a label function, we are going to omit it for simplicity. The distributed languages under study in this work are the following
\begin{align*}
\interval &= \{(G,\id)\colon G\text{ is a interval graph}\}\\
\chordal &= \{(G,\id)\colon G\text{ is a chordal graph}\}\\
\circulararc &= \{(G,\id)\colon G\text{ is a circular-arc graph}\}\\
\permutation &= \{(G,\id)\colon G\text{ is a permutation graph}\}\\
\trapezoid &= \{(G,\id)\colon G\text{ is a trapezoid graph}\}
\end{align*}

\subsection{Proof Labeling Schemes}

Formally, we define a \emph{proof-labeling scheme} (PLS) for a distributed language $\mathcal{L}$ as a pair consisting of a prover and a verifier.  The \emph{prover} is an untrusted oracle that, given a network configuration $(G,\id)$, assigns a \emph{certificate} $c(v)$ to each node $v$ the graph. The \emph{verifier} is a distributed algorithm that runs locally at each node $v$ in $G$. This verification algorithms demand first to each node $v$ to communicate with its neighbors $w\in N_G(v)$, sending $c(v)$ (and possibly its \id's) and receiving the certificates  (and possibly its \id's) from all its neighbors. Given $\id(v)$, $c(v)$, and the certificates and \id's given by its neighbors, each node $v$ runs the verification algorithm with this information to output either accept or reject.

A PLS is considered correct if it satisfies the following two conditions:

\medskip

 \begin{itemize}
 \setlength\itemsep{0em}
 \item Completeness: If $(G,\id) \in \mathcal{L}$  then the prover can assign certificates to the nodes such that the verifier accepts at all nodes,
 
 \medskip
 
\item Soundness:  If $(G,\id) \notin \mathcal{L}$ then, for every certificate assignment to the nodes by the prover, the verifier rejects in at least one node.
\end{itemize}

\medskip

The complexity measure of a PLS is the \emph{proof-size} $f(n)$, measured in function of the number of nodes $n$, and defined as the maximum length of any message sent by the prover to the nodes or between neighbors in all network configurations $(G,\id)$ with $n$ nodes.

\subsection{Toolbox} \label{sec:tool}
In this subsection, we present already established protocols which we are going to use, in this paper, as subroutines. Note that some of these subroutines solve problems that are not decision problems.

\subsection{Spanning Tree and Related Problems} The construction of a spanning tree is a fundamental component for various protocols in the PLS model. Given a network configuration $\langle G, \id \rangle$, the {\spanningtree} problem involves creating a spanning tree $T$ of $G$, with each node possessing knowledge about which of its incident edges are part of $T$. 

\begin{proposition}
	\cite{korman2010proof}
	\label{proto:spanningtree}
	There is a PLS for {\spanningtree}  with proof-size of $\cO(\log n)$ bits.
\end{proposition}

It should be helpful, for understanding the PLS model, to show a PLS for verifying a spanning tree.

\begin{protocol}\label{proto:st}
	First, the prover gives to each node $v\in V$ the following information.
	
	\begin{itemize}
		\item The identifier of the root $r\in V$ of the spanning tree.
		\item The identifier {$p_v$} of its father in the tree.
		\item Its distance $d(v)$ and the distance of its father $d(p_v)$ to the root $r$.
	\end{itemize}
	
	Then, in the verification round, each node $v\in V$ verifies whether
	\begin{itemize}
		\item All the nodes received the same root $r\in V$.
		\item The \id of its father $p_v$ is the \id of some neighbour.
		\item If $d(p_v) = k$, then $d(v) = k+1$. 
	\end{itemize}
	
	Each node accepts only if all three conditions are satisfied; otherwise, it rejects.
\end{protocol}

Now let us analyse the correctness and soundness of the protocol.\\

{\bf Correctness.} An honest prover provides a unique root and the correct distances in the tree, so all nodes accept.\\

{\bf Soundness.} If the prover gives two or more different roots, then the nodes reject because there are two neighbors $u,v$ with different root nodes. If the tree given by the prover forms a cycle, then there exist two nodes $u$ and $v$ such that $u$ is the parent of $v$ but $d(v)<d(u)$, and $v$ rejects. Therefore, the tree constructed by the prover has to be correct, and thus the distances too.\\

{\bf Proof-size analysis.} As node identifiers can be encoded in $\cO(\log n)$ and the maximum distance in an $n$-node graph between two nodes is $n-1$, it follows that the distances $d(v)$ can also be encoded in $\cO(\log n)$.

From the protocol of Proposition~\ref{proto:spanningtree}, we can construct another protocol for the $\sizeofG$ problem. In this problem, the nodes are given an input graph $G=(V,E)$ and must verify the exact value of $|V|$, assuming that the nodes only know a polynomial upper bound on $n=|V|$. Proposition~\ref{prop:numbernode} states that there exists a PLS for $\sizeofG$ with certificates of size $\cO(\log n)$.

\begin{proposition}\cite{korman2010proof}\label{prop:numbernode}	
	There is a PLS for $\sizeofG$ with certificates of size $\cO(\log n)$.
\end{proposition}

\begin{protocol}\label{proto:numbernode}
	In the first round, the prover gives to each node $v\in V$ a certificate with the following information
	\begin{itemize}
		\item The information needed according to \Cref{proto:st} to construct a valid spanning tree $T$.
		\item The number of nodes $c_v$ in $T_v$, the T-subtree rooted in $v$.
	\end{itemize}
	
	In the verification round each node $v\in V$ validates that the spanning tree constructed is correct according to the verification round of \Cref{proto:st} and that
	\[c_v = 1 + \sum_{\substack{\omega\text{ children}\\\text{ of }v}}c_\omega. \] 
	
\end{protocol}

Soundness and completeness follow directly. Notice that nodes can check with their neighbours that they all received the same $n$, and the root checks whether this value is correct.

For two fixed nodes $s,t \in V$, problem $\stpath$ is defined in the usual way:  given a network configuration $\langle G, \id \rangle$, the output is a path $P$ that goes from $s$ to $t$. 
In other words,  each node must end up knowing whether it belongs to $P$ or not; and, if it belongs to the path, it has to know which of its neighbors are its predecessor and successor in $P$.

	\begin{proposition}\cite{korman2010proof}
		\label{prop:stpath}	
		There is a PLS  for $\stpath$ with certificates of size $\cO(\log n)$.
	\end{proposition}

\begin{protocol}
	
	The prover sends to each node $v$ a bit $b_v\in\{0,1\}$ which reports if the node is part of the path ($b_v = 1$) or not ($b_v = 0$).
	
	If $b_v = 1$, the prover also sends the identifiers of its predecessor and successor in the path. 
	
	In the verification round. each node $v$ such that $b_v = 1$ and $v\not =s,t$, verifies that exactly two neighbours are part of the path, and exactly one neighbour has $v$ as predecessor and one neighbour has $v$ as successor. 
	In the same way, $s$ and $t$ verify that they have one successor and predecessor, respectively. 
\end{protocol}

Based on the aforementioned results, we assume the existence of PLSs with logarithmic proof-size for computing \spanningtree,  \sizeofG\ and  \stpath, throughout the paper. We treat these algorithms as black boxes and employ them as subroutines in our protocols.

\section{Interval and Chordal Graphs} \label{sec:interval}

We begin with the study of a PLS to recognize the class of interval graphs. An interval graph is a graph $G=(V,E)$ where each node $v \in V$ can be identified with a unique interval $I_v$ on the real line such that $uv \in E \iff I_u \cap I_w \neq \varnothing$. An example can be seen in Figure~\ref{fig:intervalEx}.

\subsection{Proper Interval Graphs} 

As a warm-up, we start with the problem of recognizing the subclass of \emph{proper interval graphs}, which are interval graphs
that can be represented by intervals in such a way that no interval properly contains another. This class is equivalent to the class of 
\emph{unitary graphs}, where all intervals have length one~\cite{roberts1969indifference}. The following proposition gives another, useful characterization of proper interval graphs.

\begin{proposition}[\cite{gardi2007roberts}]\label{lem:proper}
	A graph $G$ admits a representation by proper intervals if and only if there exists an ordering  $\{v_i\}_{i=1}^n$ such that:
	\[	\forall i, j,k: i < k < j, \quad  v_i v_j \in E \Longrightarrow  v_i v_k \in E \text{ and } v_k v_j \in E. \]	
\end{proposition}
The core idea for this class and those that follow is to represent the class property as a combination of simpler instructions to be shared among all nodes. As making the adjacency of a node explicit would require proofs to be too large, we intend to exploit the geometric properties of these classes in order to represent their adjacency with a constant amount of \textit{log}-sized labels.

 Considering that constructing a PLS for recognizing proper interval graphs by use of the previous characterization is rather direct

\begin{theorem}\label{thm:properInter}
There is a PLS for $\propInterval$ using certificates of size $\mathcal{O}(\log n)$ bits.
\end{theorem}

\begin{proof}
	
The certificate that the prover sends to each node $v \in V$ has two parts.

\smallskip

\begin{itemize}

\item A number $i_v \in [n]$, which will be interpreted as the position of node $v$ in the ordering.

\item Two different $\id$s: $\id_{first}$ and $\id_{last}$, for having a global consensus on the first and last node of the ordering.
		
\end{itemize}

\smallskip

The algorithm performed by the nodes is as follows. 
Each node $v$ interprets  $i_v$ as its position in the ordering.  Then, the nodes check locally that they all received the same $\id_{first}$ and $\id_{last}$.
The nodes with these $\id$s  correspond to the first and last nodes of the ordering, which checked that they received numbers $1$ and $n$.
Finally, every node $v$ checks locally that all its $d_v$ neighbours receive different numbers in $[i_v-a,i_v-1] \cup [i_v+1,i_v+b]$ with $a+b=d_v$ and with $[k,k-1]= \emptyset$.
The first node checks that $a=0$ and the last node checks that $b=0$. All the other $a$'s and $b$'s must be different than zero.

\smallskip

{\bf{Completeness.}} Suppose first that the graph $G$ is a proper interval graph. Since an honest prover provides the correct ordering, 
from Proposition~\ref{lem:proper} it follows that every node accepts.

\smallskip

{\bf{Soundness.}} Now we are going to prove that, if every node accepts, then the graph $G$ is a proper interval graph. 
First, the nodes check that the assignment of numbers corresponds to an ordering.
Note that every node $v$ must have a neighbour $v'$ with $i_{v'}=i_{v}+1$, with the exception of the last one. Since there is one node that receives a 1 and another that receives an $n$, then every node must receive a different number in $[n]$.
Therefore, the assignment given by the prover is indeed an ordering from $1$ to $n$. We can denote the nodes as $v_1,\ldots v_n$. Now, we need to prove that the ordering satisfies the condition of Proposition~\ref{lem:proper}. In fact, let $i<j<k$ such that $v_i v_j \in E$. From the point of view of $v_j$, since $v_i$ is a neighbour, then $v_k$ must also be a neighbour (because $[k,j] \subseteq[i,j])$. On the other hand, from the point of view of $v_i$, since $v_j$ is a neighbour, then $v_k$ must also be a neighbour (because $[i,k] \subseteq[i, j]$).
\end{proof}

\subsection{Chordal Graphs and the Particular Case of Interval Graphs}\label{chordalinterval}

	We now extend this strategy of representing the adjacency of nodes in a \textit{compact} manner to a more general setting by studying the classes of Chordal and Interval graphs. Both of these classes
	are closely related and share similar challenges when it comes to distributing the proof of its structure among all the nodes in the graph.

Chordal graphs are intersections of subtrees of a tree. More precisely, $G$ is chordal if and only if there exists a tree $T$ such that every node of $G$ can be associated with a subtree of $T$  in such a way that two nodes of $G$ are adjacent if their corresponding subtrees intersect.  The name chordal comes from the fact that a graph is chordal if and only if every cycle of length at least 4 has a chord. That is, for any integer $k\ge 4$, $G$ does not have a  $C_k$ as an induced subgraph. These graphs are especially relevant from an algorithmic perspective as several graph properties (finding the largest independent set or clique, computing the chromatic number, etc) can be efficiently computed when the input graph is restricted to this class~\cite{hoang1994efficient, rose1976algorithmic}.

As for interval graphs, the structure of a chordal graph is determined by its maximal cliques: while an interval graph can be represented as a path formed by its maximal cliques, in the case of chordal graphs these can be seen as trees.

A tree decomposition of some graph $G$ is a tree $T_G$ where each node $b\in T_G $ (referred to as \emph{bags}) represents a set of nodes $b\subseteq V$ in the original graph with the following properties: {(1)} each node $v\in G$ is present in at least one bag, {(2)} for every edge $e=uv \in E(G)$ there exists a bag $b$ that contains both $u$ and $v$ and {(3)} if we define $T_v$ as the set of bags in $T_G$ to which $v$ belongs, they form a connected subgraph of $T_G$. From this, we can define a \emph{clique-tree} of a graph $G$ as the special case of a tree decomposition for $G$ where each bag represents a maximal clique of $G$.

Similarly, we define a path decomposition of a graph $G$ as a tree decomposition when the tree $T_G$ in question is a path~\cite{bodlaender1998partial}. Following the previous notation, we define a \emph{clique-path} of a graph $G$ as the special case of a path decomposition for $G$ where each bag represents a maximal clique of $G$ 
(see Figure~\ref{fig:trim}).

\begin{proposition}[\cite{brandstadt1999graph}]
	A graph $G$ is said to be chordal if and only if it admits a clique-tree.
\end{proposition} 
  \begin{proposition}[\cite{brandstadt1999graph}]
  	A graph $G$ is said of be an interval graph if and only if it admits a clique-path.
  \end{proposition}

Our goal is to show a PLS for recognizing chordal graphs. For this, we would like to find a way to simulate the nodes in the clique-tree by choosing a set of leaders for each maximal clique. Then, we would like to label each node with the range of cliques it belongs to. The problem is that, while interval graphs require only two endpoints to represent such a range, in the case of chordal graphs we would need to encode an entire subtree which would require labels of size $\cO(n\log n)$. Therefore, we need to find a more succinct way to encode a tree. 

For this, we show that we can ``trim" the graph by sequentially removing nodes by using the tree structure. If we consider a clique-tree 
rooted at some node $\rho_T$ and trim the graph in a series of $d$ steps (which are performed simultaneously by the nodes), with $d$ the depth of the clique-tree, then we can partition the set of nodes as follows.
At each step $i$, we assume that the clique-tree $T_G$ has depth $i$ and look at the leaves at the deepest level in the clique-tree, and, for each such leaf $b$, delete all the nodes which belong {\emph{only}} to this bag and name $F_b$ such a set. 
Then, we continue to step $i-1$ (where our new tree has depth $i-1$) and repeat this process. We know this set is non-empty by the maximality of the clique represented by the bag $b$. As this goes on for $d$ steps, we have that a node $v$ is eliminated from the clique-tree at the step corresponding to the lowest depth of a bag containing the node $v$, as it can be seen in Figure~\ref{fig:chordaltrim}.

\begin{lemma}\label{lem:cordalTrim}
	For any chordal graph $G$ such that $T_G$ is a clique-tree rooted at some bag $\rho_T$, consider $\{M_b\}_{b\in T_G}$ to be its set of maximal cliques. Then, it is possible to partition the nodes in $V$ into  a family $\{F_b\}_{b \in T_G}$ such that for any pair of bags $b\neq b'$ with $\mathsf{depth}(b) \ge \mathsf{depth}(b')$ it holds that $F_b\subseteq M_b\setminus M_{b'}$.
\end{lemma}

	\begin{figure}[!ht]
	\centering
	\includegraphics[width=0.6\linewidth]{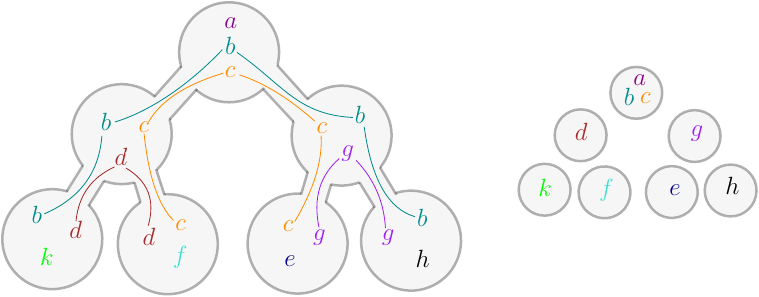}
	\caption{Graph partition according to Lemma~\ref{lem:cordalTrim}. Given a tree decomposition, each node $v$ is positioned at a different set depending on the bag containing $v$ such that its depth is the lowest in the clique-tree.}
	\label{fig:chordaltrim}
\end{figure}
\begin{proof}
	We show this by induction on the number of bags in the clique-tree of a graph $G$, given by $|T_G|$. Indeed, if $T_G$ is composed of only two bags $b$ and $b'$  with $T_G$ rooted at $b$ then, as both $M_b$ and $M_{b'}$ are maximal cliques, we simply consider $F_{b'}$ to be $M_{b'}\setminus M_b$  and $F_b= M_b$. Clearly. these sets are disjoint.

	Consider now $|T_G|=k$ with $k\ge 3$. Then, as $T_G$ is a tree, there must exists a bag $b \in T_G$ with degree 1, with $b'$ its parent in $T_G$.  We then have that $F_b = M_b \setminus M_{b'}$ is disjoint with all other bags in the tree $T_G$ as, otherwise, if there exists a node in $F_b$ that is also in a bag at a lower depth then, by the definition of a clique-tree, it must belong to $M_{b'}$. From there, we have that the tree $T' = T_G'- b$ is a clique-tree for the graph $G' = G-F_b$. It follows, by induction, that there exists a disjoint collection $\{F_{\bar{b}}\}_{\bar{b}\in T_G-b}$ with the above properties. Hence, the family $\{F_{\bar{b}}\}_{\bar{b}\in T_G-b} \cup \{F_b\}$ is as desired.
\end{proof}

Now, we would like to select a collection of leaders from each bag in $T_G$ and provide certificates to these leaders in order to verify the overlaying clique-tree. Two difficulties arise:
\begin{enumerate}
	\item The leaders in each pair of adjacent bags in $G$, may not be connected, as by construction the leader in some bag $b$ does not necessarily belong to $b$'s parent. Therefore, we would like to consider a collection of leaders (in the intersection of adjacent bags) in order to simulate $T_G$'s edges.
	\item If we could solve the first problem, we have no guarantees that we will be able to choose a leader for each edge of $T_G$ in an injective manner: it could be the case that a leader belongs to the $\Omega(n)$ maximal cliques adjacent to the same bag, and should therefore handle too many messages.
\end{enumerate}

To solve the first issue, we show how to choose a root for $T_G$ and a collection of nodes that belong to the intersection of adjacent bags in such a way that we are able to verify the correctness of the tree structure.
\begin{lemma}\label{lem:cordal}
	Given a chordal graph $G$, there exists a rooted clique-tree  $T$ such that it is possible  to choose a collection of leaders for each bag $\{v_b\}_{b\in T}$ and auxiliary nodes $\{w_\ell\}_{\ell\in T}$ for each leaf in $T$ such that if $\mathsf{depth}{(b)}$ is the depth of a bag $b$ in the tree and $t(b)$ is the parent of the bag $b$ in $T$, then:
	\begin{itemize}
		\item For each $b\in T, \quad v_b \in b$.
		\item For each $b\in T$, $v_{b}v_{t(b)}\in E(G)$. 
		\item If $\mathsf{ depth}{(b)} \neq\mathsf{ depth}{(b')}$, then $v_b \neq v_{b'}$.
		\item If $b\in T$ is a leaf, then $w_bv_b \in E(G)$.
		\item$ \{v_b\}_{b\in T} \cap \{w_\ell\}_{\ell\in T} = \emptyset$.
	\end{itemize}
\end{lemma}
\begin{proof}
	Let $T$ be a  rooted {clique-tree} in $G$ such that its set of leaves is the largest and the sum of their depths is as small as possible. Consider now $b_{\rho}\in T$ to be the root of $T$, and let some arbitrary node $v_r \in b_{\rho}= b^0$ be its leader. If we define $\{b^1_i\}_{i=1}^{\ell}$ to be the children of $b_{\rho}$ in $T$, for each $i$ we choose an arbitrary leader $v^1_i$ in $b_r\cap b^1_i$, which is non empty as $G$ is connected.
	
	Consider $b^j$ to be any node at level $j\ge 1$ of the tree with $b^{j-1}$ its parent, with $v^j$ its leader and $v^j \in b^j \cap b^{j-1}$.
	\begin{itemize}
		\item If $b^j$ is a leaf, we choose an auxiliary node $w_j \in b^j$ with $w_j \notin b^{j-1}$. We can pick such a node because $b^j$ represents a maximal clique in $G$ and, otherwise, it would be contained in $b^{j-1}$.
		\item If $b^j$ is not a leaf, let $\{b^j_i\}_{i=1}^{k}$ be the set of  $b^j$'s children. For each $i$ we choose a leader for $b^j_i$ in $b^j \cap b^j_i \setminus b^{j-1}$ as, otherwise, we would have that $b^j\cap b^j_i \subseteq b^{j-1}$ for some $i$. This implies that $b^j_i \cap b^j \subseteq b^{j-1}\cap b^j_i$ and it would be possible to define a  new tree $T'$ where $b^j_i$ is a child of $b^{j-1}$ instead of $b^j$. Now, if $b^j_i$ was not a leaf, we would have a tree with more leaves than $T$. Otherwise, in case that $b^j_i$ were a leaf, the sum of the depths of each of its leaves decreases by one. This contradicts our choice for $T$.
		
		Hence, we can choose a leader $v^j_i$ in $b^j_i \cap b^j$ for each value of $i$ and then go for the next level.
	\end{itemize}
	From the construction, we have that for any pair of adjacent bags in $T$, either their leaders match or are adjacent, as both belong to the intersection of their respective bags.	Also, by the way, the leaders were chosen in the latter point, by choosing bags at different depths, it follows directly that their leaders must be different.
\end{proof}
\begin{figure}[!ht]
	\centering
	\includegraphics[width=0.6\linewidth]{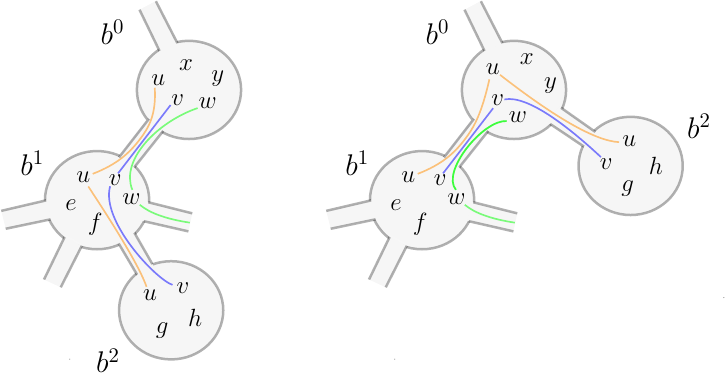}
	\caption{Transition from a tree decomposition to another one by reassigning a leaf at an upper level, in case that a bag intersection ($b^1  \cap b^2 =\{u,v\}$) is contained in an intersection at an upper level ($b^0\cap  b^1 = \{u,v,w\}$). We can do this while keeping a feasible decomposition.}
	\label{fig:colgante}
\end{figure}

Now we have a set of leaders who belong to the intersection of bags at different levels of a rooted clique-tree. Yet, these leaders may have several bags at the same level assigned to them, as multiple bags may share a unique element at the intersection with the previous level. 

To solve this issue we simply need to make use of the tree structure by setting, for each leader of a bag $\rho_b$, the certificate provided by one of its children (e.g. the one with the smallest identifier). If a node is the leader of several bags, 
it still receives all corresponding proofs after these are exchanged at the verification round, with auxiliary nodes being chosen in order to cover the case when a node is the leader of multiple leaves in the tree. Hence, we have that each bag leader will receive a unique message, no matter the number of bags it represents. Now we are ready to prove the theorem.

\begin{theorem}\label{thm:chordal}
There is a PLS for $\chordal$ using certificates of size $\mathcal{O}(\log n)$.
\end{theorem}

\begin{proof} Assuming that $G$ is chordal, we first select a collection of nodes that represents a bag $b$ in $T_G$. We do this by selecting a single element $\rho_b$ from each set $F_b$ according to Lemma~\ref{lem:cordalTrim}, as well as a spanning tree triple for simulating the overlaying clique-tree, which we denote by $\langle \id(\rho_T), d_T(v), t_T(v)\rangle$ indicating the unique leader for the root of the clique tree, as well as the distance to it and $\rho_b$'s parent in this structure. 
	
	The prover provides each node $v\in F_b$ with
	\begin{itemize}
		\item The size of the clique tree $|T_G|$, along with the $\id $ of the leader for its root $\rho_T$.
		\item A label $F(v)$ corresponding to the identifier $\id(\rho_b)$ of the leader in the set $F_b$ to which a node $v$ belongs to, as well as the size of $F_b$.
		\item The distance from the bag $b$ to the root $\rho_T$  in the clique tree with $v\in F_b$, given by $\textsf{depth}(v)$.
	\end{itemize}
	
	\bigskip
	
	Also, in order to verify the tree structure, we choose a collection of leaders $\{e_{bb'}\}$ for each edge $bb' \in E(T_G)$ according to Lemma~\ref{lem:cordal}, as well as the corresponding auxiliary nodes to pass these messages.
	Then, the nodes exchange their messages and they check the following:
	
	\bigskip
	
	\begin{enumerate}
		\item The collections $\{\rho_b\}_{b\in T_G}$ and $\{e_{bb'}\}_{bb' \in E(T_G)}$ verify in conjunction the correctness of the tree structure.
		\item There is a unique root $\rho_T$.
		\item If $v \in F_b$, then $v$ checks that the nodes with $F(u)= \id(\rho_b)$ form a clique.
		
		\item If $v$ and $u$ are adjacent with $\textsf{depth}(v) \leq \textsf{depth}(u)$, then $v$ is  adjacent to all the leaders $\rho_b$ (and their sets $F_b$)  in the unique path between $F(v)$ and $F(u)$ which also coincides with the unique path between $F(u)$ and $\rho_T$. In particular, if $\textsf{depth}(v) = \textsf{depth}(u)$, then they must have the same leader.
	\end{enumerate}
	If all the previous conditions hold, then all nodes accept. Now, we check the correctness of this protocol.
	
	\smallskip
	
{\bf{Completeness.}} Suppose first that the graph $G$ is chordal. An honest prover will provide each node $v$ with its correct set $F_b$ according to the underlying clique tree $T_G$ which all leaders check correctly. By the definition of the clique tree, it follows that no node has a neighbour at the same depth from a different bag. As the set of bags to which $v$ belongs to corresponds to a connected subgraph, it follows that if $v$ is in a bag $b$ and it is connected to a node $u$ (which belongs to a bag $b'$) at a larger depth, then it is adjacent to all nodes in the bags (and therefore the sets $F_{b'}$) in the path between $b$ and $b'$. With this, each node $v$ recognizes that its sets $F_b$ are a clique and that the depth for the set of each of its neighbours is consistent with its set. Therefore, all nodes accept.

\smallskip
		
{\bf{Soundness.}} Suppose now that the graph $G$ is not chordal. We have that, by the constructions in Lemmas~\ref{lem:cordalTrim} and \ref{lem:cordal}, the leaders chosen for both the bags and the edges between them can correctly verify the structure of the clique tree. Now, suppose that $G$ has an induced cycle $\{v_!, \dots v_k\}$ with its nodes arranged such that $v_1$ is the node of largest depth. It must be that at least one of them has a different depth from the rest as otherwise they would reject either because their leaders are different, or because the corresponding set $F_b$ should be a clique and there exist at least two non-adjacent nodes.  Suppose that $\textsf{depth}(v_1)=i$, $\textsf{depth}(v_k)=j$ and $\textsf{depth}(v_2)=k$ with $i > j\ge k$. Then, it must be that $v_k$'s leader lies in the unique path between $\rho_T$ and $v_1$'s leader as otherwise it would notice an inconsistency with $v_1$'s proof and it would reject. This must also be true for $v_2$ for being at a smaller depth. Then, it must lie in the path between $v_1$'s leader and $\rho_T$ and, therefore, adjacent to $v_k$'s leader and subsequently to $v_k$ itself, which contradicts the fact that they are not adjacent as they belong to a large induced cycle.
\end{proof}

As a corollary, we obtain a PLS  for the problem \interval\ by considering the fact that, as described above, interval graphs are a particular subclass of chordal graphs where the clique-tree corresponds to a path~\cite{bodlaender1998partial}. From here it suffices to repeat the same protocol while each leader (with the exception of the root leader) additionally verifies that any bag assigned to it has unique children in the clique tree $T_G$.

\begin{corollary}\label{cor:interval}
There is a PLS for $\interval$ using certificates on $\mathcal{O}(\log n)$ bits.
\end{corollary}

\section{Circular Arc Graphs}\label{sec:circular}

Circular arc graphs are a natural extension of interval graphs. Indeed, they are the graphs that admit a representation by arcs on a circle, and appear, for instance,  in the study of resource allocation problems for periodic tasks~\cite{mandal2006maximum}. We study this class of graphs as we wish to check whether previous results can be extended to this new setting without a large increase in the proof-size.
We start by formally defining this new class and, again, studying two variations: where no pair of arcs are properly contained and then the general case. For the sake of simplifying the notation, we identify the set of \id's as $[n]=\{0, \dots n-1\}$. We say that a graph $G=(V,E)$ admits a circular arc representation if there exists a family of arcs in the unit circle $\{A_v\}_{v\in V}$ such that the adjacency of $G$ can be determined by the intersection of arcs. That is \[ \forall v \in V, \: \exists A_v: \quad \forall u,w \in V, uw \in E \Longleftrightarrow A_u \cap A_w \neq \varnothing\]

We say that some graph $G$ is a \emph{proper} circular arc graph if it admits a representation where no arc is contained in another.

\subsection{Proper Circular Arc Graphs}\label{proper_circular}

As in previous proofs, the main question is how to represent the adjacency of a node in a succinct manner, considering the geometric properties of this class. We proceed as follows.  We assume that the \id's are ordered counter-clockwise. We say that, given $i<j$, the adjacency of a node is given either by $(i,j)$, which we define as $\{i, i+1, \dots j\}$ or $(j,i)$, which we set to be $\{ j, j+1, .. n-1, 0, ..i\}$. If a graph $G$ satisfies this property we say that its augmented adjacency matrix (the adjacency matrix of $G$ with the addition of 1's in the diagonal), denoted by $M^*(G)$, has the \emph{circular} 1's property~\cite{tucker1970}. 

Now, again, we can have graphs that follow this property while allowing a representation by proper circular arcs. So there must be another property that we need in order to pin down this graph class.
Fortunately, given a characterization by Tucker~\cite{tucker1970}, we can show how to recognize this class with a single round of interaction. For this, we start by giving some definitions for symmetric matrices.

First, consider $\pi$ to be a permutation of $[n]$ and some matrix $M$. The matrix $M_\pi$ is obtained when both the rows and columns of $M$ are reordered according to $\pi$. Second, consider a symmetric $\{0,1\}$-matrix $M$ with $1$'s in the diagonal and the circular 1's property. Then, we define $\textsf{last}[M,j]$ to be the largest value $i$ such that $M_{i,j} = 1$ and $M_{i+1,j}=0$. If such an $i$ does not exist (meaning the column $M_{\cdot, j}$ has only 1 entries) we set $\textsf{last}[M,j]=\bot$.

Last, for the sake of notation, consider $\sigma_{\textrm{inv}} :[n]\to [n]$ to be the permutation given by $i \to n-i$ if $i\neq n$ and $n$ otherwise and $\sigma_{\textrm{sh}} :[n]\to [n]$ to be the permutation given by $i \to i+1$ if $i\neq n$ and $1$ otherwise.

\begin{definition} Given a symmetric $\{0,1\}$ matrix $M$ with 1's in the diagonal, we say that it has circularly compatible 1's if $M$ has the circular 1's property and, for any reordering $\pi$ of the rows (and respective columns) of the matrix constructed by a finite composition of $\sigma_{\mathrm{inv}}$ and $\sigma_{\mathrm{sh}}$,  it follows that $\mathsf{last}[M_{\pi}, 0] \leq \mathsf{last}[M_{\pi}, 1]$, unless one of these values is $\bot$.
\end{definition}

With these definitions, we can finally describe the characterization given by Tucker for this class of graphs.

\begin{proposition}[\cite{tucker1970}]
	\label{prop:proper_circ}
	A graph $G$ is a proper circular arc graph if and only if its nodes admit an ordering $\{\pi_v\}_{v\in V}$ such that its augmented adjacency matrix $M^*(G)$ has the circularly compatible 1's property.
\end{proposition}

This characterization suits us greatly as its condition is highly local. If we were able to find such an ordering, then every node would only need to verify it by checking the previous and next nodes in the ordering. For this, we note two important remarks.

\begin{observation}[\cite{tucker1970}]\label{obs:proper_circ}If we sort the nodes according to their right endpoint in counter-clockwise order, we have that the nodes adjacency matrix has the circularly compatible 1's property according to this ordering.
\end{observation}

\begin{observation}
	We can rotate the arcs in a graph in such a way that, if each node $v$ is sorted according to the previous order $\pi$, it follows that $v$ is adjacent to $\pi_v-1$ and $\pi_v+1$ modulo $n$, with the exception of the last node (in position $n$) which may not be connected to the first.
\end{observation}

Now we can start to describe the protocol.

\begin{theorem}\label{prop:properCirc}
There is a PLS for $\propCircular$ using certificates of size $\mathcal{O}(\log n)$.      
\end{theorem}
\begin{proof}
	
	Assume that the prover assigns to each node a pair $A_v= ( r_v, \ell_v)$ which correspond to $v$'s arc coordinates when the arc is visited in a counter-clockwise direction. Here we assume that such coordinates are given as a value between $(0, 2\pi)$ with any pair of values being at a distance at least $1/\poly(n)$ from each other (and, as such, we require $\cO(\log n)$ bits to represent such a range). Now, let $v_1$ be the node such $r_v$ is the smallest and such that, in case the graph is a proper circular graph, is the first node if we sort them according to their right endpoint as in Proposition~\ref{prop:proper_circ} and Observation~\ref{obs:proper_circ}.

	Now,  first, we ask the prover to provide the identifier $\id(v_1)$ of such a node, as well as proof that it is the only node with the smallest right endpoint, which can be provided by sending a spanning tree triple $\tree$ as well as a verification through the spanning tree that $v_1$ is the unique node with the smallest value for $r_v$. We also ask the prover to provide the size of the graph $n(G)$ which can also be verified through the spanning tree.
	
	Next, we ask the prover to send to each node a position $\pi_v$ such that $\pi_v = i$ means that $r_v$ is the i-th largest value for a right endpoint in counter-clockwise order. Also, we ask the prover to provide each node with a range $(v_{\min}, v_{\max} )$ which correspond to the positions in $\pi$ such that $v$'s adjacency equals the set of nodes whose positions are $\{ v_{\min}, v_{\min}+1 , \dots v_{\max}\}$ given as a circular sequence as described previously.
	
	Then, at the verification round, all nodes exchange these messages, verifying the existence of a node $v_1$ by using the spanning tree and, starting from $v_1$, each node $v$ with $\pi_v= i $  checks that there is a unique node labeled by $i+1$ whose arc intersects with its own and whose left endpoint is immediately after his. They also check that their adjacency is circular, meaning that it has a neighbor labeled with each position in the range $(v_{\max}, v_{\min}),$ with arcs consistent with such an order. 
	
	Finally, in order to check that the matrix has the circularly compatible $1$'s property, they do the following. In order to check the first two columns in each permutation obtained by a composition of $\sigma_{\textrm{inv}}$ or $\sigma_{\textrm{sh}}$ we simply ask each node to adjust its range according to these permutations, such that each node $v$ positioned at $\pi_v$ with neighbors $w$ and $u$ such that $\pi_w= \pi_v-1$ and $\pi_u= \pi_v+1$ must simply consider two cases: {(1)} When $v$ is first and $u$ is second, which occurs when we shift $\pi$ (by applying $\sigma_{\textrm{sh}})$ until $v$ is first in the order or {(2)} when $v$ is first and $w$ is second, which occurs when we invert the order by applying $\sigma_{\textrm{inv}}$ and then shift $\pi$ until $v$ is first. In this way, the verification of all $2n$ possible permutations is distributed among the nodes, with each $v$ in charge of two cases.
	
	We explain how to handle both cases by adjusting the range of $v$ and that of its neighbours as follows:
	\begin{itemize}
		\item If $v$ is first and $u$ is second, we can obtain the corresponding range by setting $k= n- \pi_v +1$ and translating both ranges by $k \mod n$  as $(\bar{v}_{\min}, \bar{v}_{\max}) = (v_{\min} + k \mod n, v_{\max} +k \mod n)$ and a similar construction for $u$.
		\item If $v$ is first and $w$ is second, first we obtain the range after applying $\sigma_{\textrm{inv}}$ as $(\bar{v}_{\min}, \bar{v}_{\max}) = ( n-v_{\max} +1 \mod n, n- v_{\min} +1\mod n)$ and then shifting $v$ to the first position by adding $\pi_v$ on both sides modulo $n$, and a similar construction for $w$.
	\end{itemize}
	
	Given these two different ranges, each node checks that (unless either itself or $w$ or $u$ are universal nodes) the range $\bar{v}_{\max} \leq \bar{u}_{\max}$ (respectively $\bar{v}_{\max} \leq \bar{w}_{\max}$), accepting if this holds and rejecting otherwise.
	
	We have that all these messages are of length $\cO(\log n)$ in one round of interaction. Therefore, it only remains to check the correctness of this protocol.

\smallskip

{\bf{Completeness.}}  Suppose that $G$ is a proper circular arc, an honest prover will provide each node with an ordering $\pi$ according to each arc's left endpoint, as well a the correct range which all nodes can verify and accept.
 
 \smallskip
 
{\bf{Soundness.}} Now, suppose that $G$ is a No-instance. From what was described above, all nodes correctly compute a starting node $v_1$ as well as the size of the graph. Also, each node $v$ with $\pi_v = i$ checks that it has a unique neighbour positioned as $i+1$, which is consistent with its arc. Combining both statements we have that all nodes must have different values in $\pi$ that match the order of their left endpoints. 

Now, if all nodes check that their adjacency is indeed circular there must exists a node $v$ which, when permuting its order such that it becomes first either $\bar{v}_{\max} > \bar{u}_{\max}$ or $\bar{v}_{\max} > \bar{w}_{\max}$, and then it would immediately reject. \end{proof}

\subsection{The General Case}\label{circular_arc}
To cope with the general case, it suffices to adapt the characterization found in~\cite{tucker1970} for this class of graphs, as it gives us a simple representation of the adjacency of each node by means of the shape of the adjacency matrix relative to a node ordering.

Given a symmetric $\{0,1\}$-matrix $M$, with ones in the diagonal, consider a column $i$ and define $U_i$ as the set of  $1's$ starting from the diagonal and going downwards in a circular manner until a zero appears. Now, define $V_i$ as the set of $1's$ on row $i$ starting from the diagonal and going rightwards in the same manner. $M$ is said to have the \emph{quasi-circular} $1's$ \emph{property} if all  $1's$ in the matrix are covered by some  $U_i$ or $V_i$.  It is important to mention that, since $M$ is symmetric, we have that $U_i$ and $V_i$ have the same size.

\begin{figure}[h!]
	\centering
	\ContinuedFloat*
	\includegraphics[width=0.6\linewidth]{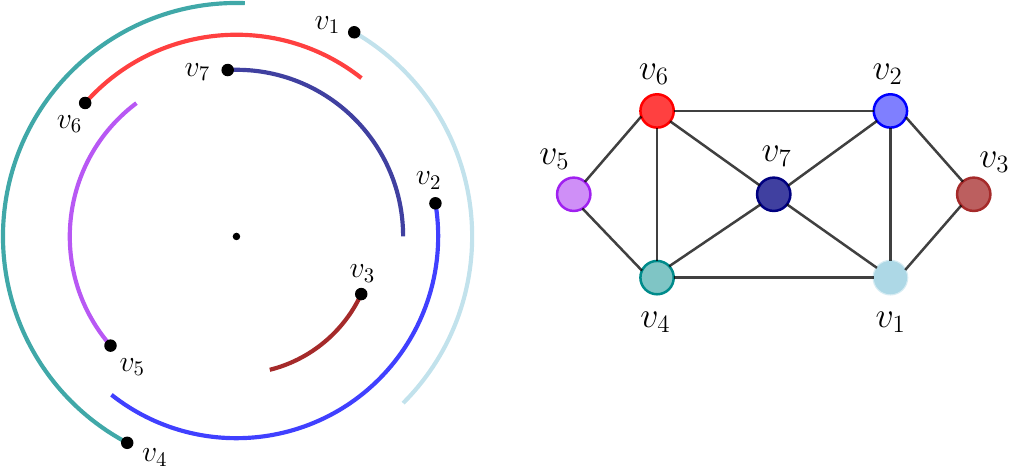}
	\caption{A circular arc representation for a graph, along with its associated drawing.}
\end{figure}

\begin{figure}[h!]
	\centering
	\ContinuedFloat
	\includegraphics[width=0.3\linewidth]{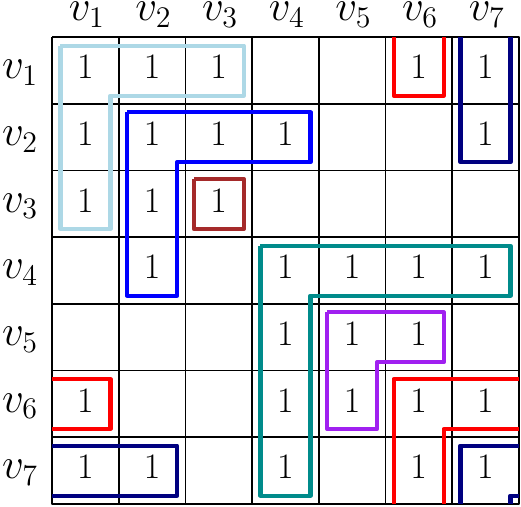}
	\caption{Augmented adjacency matrix for the previous graph with the quasi-circular  $1's$ property.}
\end{figure}

\begin{proposition}[\cite{tucker1970}]
	Let $M^*(G)$ be the augmented adjacency matrix of $G$. We have that  $G$ is a circular arc graph if and only if there exists an ordering for the nodes such that $M^*(G)$ has quasi-circular~$1's$.
\end{proposition}

From here, we can describe a PLS with cost $\cO(\log n)$.

\begin{theorem}\label{prop:circular}
There is a PLS for $\circularArc$ using certificates on $\mathcal{O}(\log n)$ bits.
\end{theorem}

\begin{protocol}
	First, the prover sends to each node $v$:
	\begin{enumerate}
		\item Its position in the ordering $\pi_v$ as well as the total number of nodes $n(G)$.
		\item A spanning tree given by the triple $\tree$.
		\item The size of its set $U_{\pi_v}$ denoted by $L_v$.
	\end{enumerate}
	
	After the nodes exchange their certificates, they check the consistency of the spanning tree and use it in order to verify that the total number of nodes is correct. Then, in order to verify the consistency of $\pi(\cdot)$ as a correct ordering, the nodes proceed as follows. 
	
	If we set $N^\pi (v)$ to be the set of nodes in $N(v)$ such that they are positioned between $\pi_v$ and $\pi_v+ L_v-1$, for $i \in \{0, \dots n-2\}$ each node $v$ in position $\pi_v=i$ must check that it has a unique neighbor $u$ positioned at $\pi_u=j$ for all positions $j$ in $\{\pi_v + h_v\}$, where $w\in N^\pi(v)$ with $\pi_w= h_v$ is the first node such that $\pi_w + L_w -1 > \pi_v + L_v -1$.
	By this process, each node starting from $i=0$ makes sure that there are nodes labeled with a position in $U_{\pi_v}$ and that there is a unique node that can continue this process after him. As $G$ is connected, we can assume that this process continues on until all nodes with positions in $\{0, n-1\}$ are verified. As there are $n$ nodes in the graph, all positions are distinct.
	
	Finally, each node $v$ with a neighbour $u$ checks that, either $u\in N^\pi(v)$ or $v\in N^\pi(u)$ and that $v$ is adjacent to all nodes  with positions in $N^\pi(v)$. Rejecting if any of these conditions are not satisfied.
	
\end{protocol}

\smallskip

{\bf{Completeness.}} We have that, if $G$ is a circular arc graph, then it admits an ordering with the previous property. Then, each node  $v$ has neighbours whose positions are between $\pi_v$ and  $\pi_v + L_v$ circularly, and any other neighbour is such that $v$ verifies that property for them. Therefore, all nodes always accept.

\smallskip

{\bf{Soundness.}} If $G$ is not a circular arc graph, then we'll have that, for any order, there exists a pair of adjacent nodes $u,v$ such that, as a pair, do not belong to any $U_i$ or $V_i$. Thus, we have that either $u$ rejects as $\pi_u \leq \pi_v + L_v-1 \mod n$ or $v$ rejects as one of them notices that fact.

\section{Trapezoid and Permutation Graphs} \label{sec:trapezoid}

Now we turn to study the class of trapezoid graphs. A graph is said to be a trapezoid graph if there exists a collection of trapezoids $\{T_v\}_{v\in V}$ with vertices in two parallel lines $\mathcal{L}_t$ and $\mathcal{L}_b$ (as in Figure~\ref{fig:trapezoidvertex}) such that $\{u,v\}\in E$ iff $T_u\cap T_v \neq \emptyset$. We call these lines \emph{top and bottom lines}. The trapezoids have sides contained in each line, and, therefore, are defined by four vertices, two in the top line, and two in the bottom line. Formally, each trapezoid $T$ is defined by the set $T = \{t_1,t_2, b_1, b_2\}$, where $t_1 < t_2$ and $b_1 < b_2$,  with $t_1, t_2 \in \mathcal{L}_t$ and $b_1, b_2 \in \mathcal{L}_b$.  

	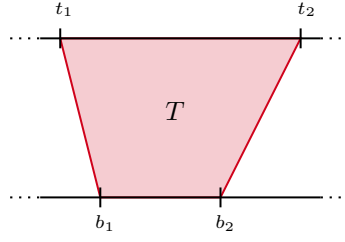
\begin{figure}[h!]
		\centering 
		\tikzset{every picture/.style={line width=0.75pt}} 

\begin{tikzpicture}[x=0.75pt,y=0.75pt,yscale=-1,xscale=1]

\draw  [color={rgb, 255:red, 208; green, 2; blue, 27 }  ,draw opacity=1 ][fill={rgb, 255:red, 208; green, 2; blue, 27 }  ,fill opacity=0.2 ] (350,45) -- (310,125) -- (250,125) -- (230,45) -- cycle ;
\draw [line width=0.75]    (220,45) -- (360,45) ;
\draw [line width=0.75]    (220,125) -- (360,125) ;
\draw [line width=0.75]    (310,120) -- (310,130) ;
\draw [line width=0.75]    (250,120) -- (250,130) ;
\draw [line width=0.75]    (350,40) -- (350,50) ;
\draw [line width=0.75]    (230,40) -- (230,50) ;
\draw [line width=0.75]  [dash pattern={on 0.84pt off 2.51pt}]  (360,125) -- (375,125) ;
\draw [line width=0.75]  [dash pattern={on 0.84pt off 2.51pt}]  (360,45) -- (375,45) ;
\draw [line width=0.75]  [dash pattern={on 0.84pt off 2.51pt}]  (205,125) -- (220,125) ;
\draw [line width=0.75]  [dash pattern={on 0.84pt off 2.51pt}]  (205,45) -- (220,45) ;

\draw (281,75.4) node [anchor=north west][inner sep=0.75pt]    {$T$};
\draw (246,132.4) node [anchor=north west][inner sep=0.75pt]  [font=\footnotesize]  {$b_{1}$};
\draw (306,132.4) node [anchor=north west][inner sep=0.75pt]  [font=\footnotesize]  {$b_{2}$};
\draw (226,25.4) node [anchor=north west][inner sep=0.75pt]  [font=\footnotesize]  {$t_{1}$};
\draw (347,25.4) node [anchor=north west][inner sep=0.75pt]  [font=\footnotesize]  {$t_{2}$};

\end{tikzpicture}
		\caption{Each trapezoid $T$ is defined by a set $T = \{ b_1, b_2, t_1, t_2\}$.}
		\label{fig:trapezoidvertex}
	\end{figure}

Consider a trapezoidal model ${T_v}{v\in V}$, as previously described. The vertices of each trapezoid can be labelled from left to right with integers from $1$ to $2n$ for both the lower and upper lines. Therefore, we can assume, without loss of generality, that the vertices defining the set ${T_v}{v\in V}$ are all distinct and have a value in the range $[2n]$. As a result, each element in the range $[2n]$ corresponds to a vertex of some trapezoid in both the top and bottom lines. 

 For $v \in V$, we call $\{t_1(v), t_2(v), b_1(v), b_2(v)\}$ the vertices of $T_v$.  Moreover, we say that the collection $\{t_1(v), t_2(v), b_1(v), b_2(v)\}$ are the \emph{vertices} of node $v$. In the following, a trapezoid model satisfying the conditions stated above is called a \emph{proper trapezoid model} for $G$. Given a graph $G=(V,E)$ (that is not necessarily a trapezoid graph), a \emph{semi-proper trapezoid model} for  $G$ is a set of trapezoids $\{T_v\}_{v\in V}$ satisfying previous conditions, such that, for every $\{u,v\} \in E$, the trapezoids $T_v$ and $T_u$ have nonempty intersection. The difference between a proper and a semi-proper model is that in the first we also ask every pair of non-adjacent edges to have non-intersecting trapezoids.  

	Given a trapezoid graph $G = (V,E)$ and a proper trapezoid model $\{T_v\}_{v\in V}$, we define the following sets for each $v\in V$:
	\begin{align*}
		F_t(v) & = \{i\in [2n]\mid i<t_1(v)\text{ and }i \in \{t_1(w), t_2(w)\} \text{ for some }w\notin N(v)\} \\
		F_b(v) & = \{i\in [2n]\mid i<b_1(v)\text{ and }i \in \{b_1(w), b_2(w)\} \text{ for some }w\notin N(v)\}
	\end{align*}

	Intuitively, the set $F_t(v)$ has the positions in the upper line to the left of $T_v$ which are vertices of a trapezoid $T(\omega)$, with $\omega\notin N(v)$. Analogously for $T_b(v)$.  We also call $f_t(v)  = |F_t(v)|$ and $f_b(v) = |F_b(v)|$. 

The Lemmas presented below provide a characterization of trapezoid graphs through equalities that can be computed locally by each node based on the information available from its neighbours.

	\begin{lemma}
		\label{lem:trapcaract}
		Let $G=(V,E)$ an $n$-connected trapezoid a graph. Then every proper trapezoid model  $\{T_v\}_{v\in V}$  of $G$ satisfies for every $v\in V$: $$f_b(v) = f_t(v)$$
	\end{lemma}

\begin{proof}
	Let $\{T_v\}_{v\in V}$ be a proper trapezoid model of $G$.  Then, given a node $v\in V$, all the coordinates in $F_t(v)$ are vertices of some $w\neq N(v)$. Such trapezoids $T_w$ have their two upper vertices in the set $\{1, \dots, t_1(v)\}$ and their two lower vertices in $\{1, \dots, b_1(v)\}$, as otherwise $T_w$ and $T_v$ would intersect. Then, the cardinality of the set $F_t(v)$ is equal to the cardinality of the set $F_b(v)$, as every position in $\{1, \dots, 2n\}$ corresponds to a vertex of some node, so if a position $j<b_1(v)$ is not in $F_b(v)$, then has to be a vertex of some neighbour of $v$. Analogous for the positions $j< t_1(v)$ in the upper line. 
\end{proof}

\begin{lemma}
	\label{lem:nontrapcarac2}
Let $G = (V,E)$ be a $n$-node graph that is not a trapezoid graph. Then, for every semi-proper trapezoid model $\{T_v\}_{v\in V}$ of $G$, at least one of the following conditions is true: 

\begin{enumerate}
    \item $\exists v\in V$ such that some value in $\{b_1(v), \dots, b_2(v)\}$ or $\{t_1(v),\dots, t_2(v)\}$ is a vertex of $\omega\notin N(v)$.
    \item  $\exists v\in V$ such that $f_b(v) \neq f_t(v)$.
\end{enumerate}

\end{lemma}

\begin{proof}

Let $G$ be a graph that is not a trapezoid graph and $\{T_v\}_{v\in V}$ a semi-proper trapezoid model. As $G$ is not a permutation graph, by definition necessarily there exists a pair $\{v,\omega\}\not\in E$ such that $T_v \cap T_\omega \neq \emptyset$. We distinguish two possible cases (see \Cref{fig:nontrapcarac}): 
\begin{enumerate}
\item  $[b_1(v), b_2(v)]_\N \cap [b_1(\omega), b_2(\omega)]_\N\neq\emptyset$ or $[t_1(v) , t_2(v)]_\N\cap [t_1(\omega), t_2(\omega)]_\N\neq\emptyset$.
\item $[b_1(v), b_2(v)]_\N \cap [b_1(\omega), b_2(\omega)]_\N=\emptyset$ and $[t_1(v) , t_2(v)]_\N\cap [t_1(\omega), t_2(\omega)]_\N=\emptyset$.
\end{enumerate}

	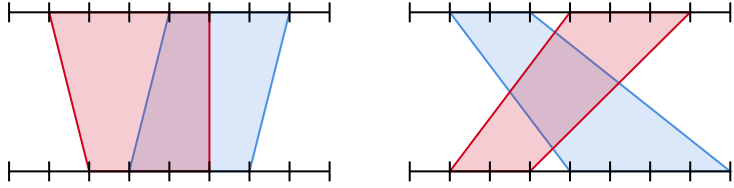
\begin{figure}[h!]
		\centering 
		\tikzset{every picture/.style={line width=0.75pt}} 

\begin{tikzpicture}[x=0.75pt,y=0.75pt,yscale=-1,xscale=1]

\draw  [color={rgb, 255:red, 74; green, 144; blue, 226 }  ,draw opacity=1 ][fill={rgb, 255:red, 74; green, 144; blue, 226 }  ,fill opacity=0.2 ] (375,60) -- (475,140) -- (395,140) -- (335,60) -- cycle ;
\draw  [color={rgb, 255:red, 208; green, 2; blue, 27 }  ,draw opacity=1 ][fill={rgb, 255:red, 208; green, 2; blue, 27 }  ,fill opacity=0.2 ] (455,60) -- (375,140) -- (335,140) -- (395,60) -- cycle ;
\draw [line width=0.75]    (315,60) -- (475,60) ;
\draw [line width=0.75]    (315,140) -- (475,140) ;
\draw [line width=0.75]    (375,135) -- (375,145) ;
\draw [line width=0.75]    (395,135) -- (395,145) ;
\draw [line width=0.75]    (415,135) -- (415,145) ;
\draw [line width=0.75]    (435,135) -- (435,145) ;
\draw [line width=0.75]    (455,135) -- (455,145) ;
\draw [line width=0.75]    (475,135) -- (475,145) ;
\draw [line width=0.75]    (315,135) -- (315,145) ;
\draw [line width=0.75]    (335,135) -- (335,145) ;
\draw [line width=0.75]    (355,135) -- (355,145) ;
\draw [line width=0.75]    (375,55) -- (375,65) ;
\draw [line width=0.75]    (395,55) -- (395,65) ;
\draw [line width=0.75]    (415,55) -- (415,65) ;
\draw [line width=0.75]    (435,55) -- (435,65) ;
\draw [line width=0.75]    (455,55) -- (455,65) ;
\draw [line width=0.75]    (475,55) -- (475,65) ;
\draw [line width=0.75]    (315,55) -- (315,65) ;
\draw [line width=0.75]    (335,55) -- (335,65) ;
\draw [line width=0.75]    (355,55) -- (355,65) ;
\draw  [color={rgb, 255:red, 74; green, 144; blue, 226 }  ,draw opacity=1 ][fill={rgb, 255:red, 74; green, 144; blue, 226 }  ,fill opacity=0.2 ] (255,60) -- (235,140) -- (175,140) -- (195,60) -- cycle ;
\draw  [color={rgb, 255:red, 208; green, 2; blue, 27 }  ,draw opacity=1 ][fill={rgb, 255:red, 208; green, 2; blue, 27 }  ,fill opacity=0.2 ] (215,60) -- (215,140) -- (155,140) -- (135,60) -- cycle ;
\draw [line width=0.75]    (115,60) -- (275,60) ;
\draw [line width=0.75]    (115,140) -- (275,140) ;
\draw [line width=0.75]    (175,135) -- (175,145) ;
\draw [line width=0.75]    (195,135) -- (195,145) ;
\draw [line width=0.75]    (215,135) -- (215,145) ;
\draw [line width=0.75]    (235,135) -- (235,145) ;
\draw [line width=0.75]    (255,135) -- (255,145) ;
\draw [line width=0.75]    (275,135) -- (275,145) ;
\draw [line width=0.75]    (115,135) -- (115,145) ;
\draw [line width=0.75]    (135,135) -- (135,145) ;
\draw [line width=0.75]    (155,135) -- (155,145) ;
\draw [line width=0.75]    (175,55) -- (175,65) ;
\draw [line width=0.75]    (195,55) -- (195,65) ;
\draw [line width=0.75]    (215,55) -- (215,65) ;
\draw [line width=0.75]    (235,55) -- (235,65) ;
\draw [line width=0.75]    (255,55) -- (255,65) ;
\draw [line width=0.75]    (275,55) -- (275,65) ;
\draw [line width=0.75]    (115,55) -- (115,65) ;
\draw [line width=0.75]    (135,55) -- (135,65) ;
\draw [line width=0.75]    (155,55) -- (155,65) ;

\end{tikzpicture}
		\caption{A representation of the two possible cases. In the first case, depicted in left, at least one vertex of a trapezoid is contained in the other. In the second case, in the right hand, the trapezoids intersect, but not in the vertices. }
		\label{fig:nontrapcarac}
	\end{figure}

Clearly, if the first case holds, then condition {1} is satisfied. Suppose then that there is no pair $\{v,\omega\}\not\in E$ such that $T_v \cap T_\omega \neq \emptyset$ satisfying the first case. Then necessarily the second case holds. Let $u$ be a node for which  exists $\omega \in V\setminus N(u)$ such that $T_u \cap T_w \neq \emptyset$. For all possible choices of $u$, let us pick the one such that $b_1(u)$ is the minimum. Then $u$ satisfies the following conditions:
\begin{itemize}
	\item[{(a)}] Exists a node $\omega\in V$ such that $\omega \notin N(v)$ and $T_{u}\cap T_\omega\neq\emptyset$
	\item[{(b)}] All nodes $\omega\in V$ such that $\omega \notin N(v)$ and $T_{u}\cap T_\omega\neq\emptyset$ satisfy that $t_2(\omega) < t_1(u)$ and $b_2(u)<b_1(\omega) $
	\item[{(c)}] None of the positions in $\{1,\dots, b_1(u)\}$ is occupied by a vertex of a node $\omega$ such that $\{u,\omega\}\notin E$ and $T_u \cap T_\omega  \neq \emptyset$. 
\end{itemize}

Observe that conditions {(a)} and {(b)} imply that $t_1(u) - f_t(u) > 0$, while condition {(c)} implies that $b_1(u) - f_b(u) = 0$. We deduce that condition {2} holds by $u$.
\end{proof}

We are now ready to define our protocol and main result regarding \trapezoid.

\begin{theorem}
	\label{theo:trapezoid}
There is a PLS for $\trapezoid$ using certificates on $\mathcal{O}(\log n)$ bits.
\end{theorem}

\begin{protocol}
	\label{proto:trapezoid}

The following is a one-round proof labeling scheme for \trapezoid\:\\

Given an instance $\langle G = (V,E), \id \rangle$, the certificate provided by the prover to node $v \in V$ is interpreted as follows.
		\begin{enumerate}
			\item  The number of nodes $n(G)$.
			\item  The vertices $b_1(v), b_2(v) , t_1(v), t_2(v)\in [2n]$ of the trapezoid $T_v$,  such that $b_1(v) < b_2(v)$ and $t_1(v) < t_2(v)$. 
			\item  Minimum position $p_v\in [n]$  in the upper line greater that $t_1(v)$ that is not a vertex of a neighbour of $v$.
			\item  Minimum position $q_v\in [n]$ in the lower line greater than $b_1(v)$ that is not a vertex of a neighbour of $v$.
			\item Paths $P_t$ and $P_b$  between the nodes with vertices $1$ and $2n$ in the upper and lower line, respectively.

\end{enumerate}

Then, in the verification round, each node shares with its neighbors their certificates. Using that information each node $v$ can compute $f_t(v)$ and $f_b(v)$, and check the following conditions: 

\begin{enumerate}[label=(\alph*)]
	\item The correctness of the value of $n$, according to protocol for $\sizeofG$.
	\item The correctness of the paths  $P_b$ and $P_t$, according to protocol for $\stpath$.
	\item The vertices of the trapezoid of $v$ are in $[2n]$.
	\item $T_v\cap T_\omega\neq\emptyset$ for all $\omega\in N(v)$. 
	\item All values in $[t_1(v)+1, t_2(v)-1]_\N$ and $[b_1(v)+1, b_2(v)-1]_\N$ are vertex of some neighbour of $v$.

	\item $t_2(v)<p_v$ and $b_2(v)< q_v$.
	\item If $\omega\in N(v)$ and $p_\omega<t_2(v)$, then $v$ verifies that $p_\omega$ is a vertex of some other neighbour.
	\item If $\omega\in N(v)$ and $q_\omega<b_2(v)$, then $v$ verifies that $q_\omega$ is a vertex of some other neighbour.
	\item $ f_b(v) = f_t(v)$.

\end{enumerate}

 \end{protocol}

We now analyze the soundness and completeness of our protocol.  \\

\noindent\textbf{Completeness.} Suppose that $G$ is a trapezoid graph. An honest prover just has to send the real number of nodes $n$, a trapezoid model $\{T_v\}_{v\in V}$ of G and valid paths $P_b$ and $P_t$ according the trapezoid model. Then, the nodes will verify ${(a)}$, ${(b)}$ by the completeness of the protocols for $\sizeofG$ and $\stpath$. Conditions ${(c)}$, ${(d)}$, ${(e)}$ ,${(f)}$, ${(g)}$ and ${(h)}$ are verified by the correctness of the model  $\{T_v\}_{v\in V}$. Condition ${(i)}$ is also verified, by Lemma~\ref{lem:trapcaract}.\\
    
\noindent \textbf{Soundness.} Suppose $G$ is not a trapezoid graph. If a dishonest prover provides a wrong value of $n$, or wrong paths $P_t$ or $P_b$, then at least one node will reject verifying {a} or {b}. Then, we assume that the prover cannot cheat on these values.  
    
    Suppose that the prover gives values $\{T_v\}_{v\in V}$ such that $\bigcup_{v\in V}\{t_1(v), t_2(v)\}\neq [2n]$. If some vertex of a node is not in the set $[2n]$, then that node fails to verify condition $(c)$ and rejects.  Otherwise, there exists $j \in [2n]$ such that $t_1(v), t_2(v) \neq j$, for every $v\in V$. If a node $\omega$ satisfies that $t_1(\omega) < j < t_2(\omega)$, then node $\omega$ fails to verify condition $(e)$ and rejects. Then $j$ is not contained in any trapezoid. As $P_t$ is correct, $j$ must be different than $1$ and $2n$. Also by the correctness of $P_t$, there exists a pair of adjacent nodes $u,v \in V$ such that $t_2(u) < j < t_1(v)$.  From all possible choices for $u$ and $v$, we pick the one such that $t_2(u)$ is the maximum. We claim that $v$ fails to check condition $(g)$. Since $j$ is not a vertex of any node, then $p_u \leq j$. If $v$ verifies condition $(g)$., then necessarily $p_u < j$. Then, there must exist a node $\omega \in N(v)$ such that $p_u = t_1(\omega)$. But since we are assuming that $j$ is not contained in any trapezoid, we have that $t_2(\omega) < j$, contradicting the choice of $u$. Then the prover needs to send values $\{T_v\}_{v\in V}$ such that 
    
    The same argument prove that $\bigcup_{v\in V}\{t_1(v), t_2(v)\}= [2n]$. The same argument also prove that $\bigcup_{v\in V}\{b_1(v), b_2(v)\} = [2n]$.
  
    Therefore, if conditions  $(a)$ - $(h)$ are verified, we can assume that the nodes are given a semi-proper trapezoid model of $G$. Since we are assuming that $G$ is not a trapezoid graph, by \Cref{lem:nontrapcarac2} we deduce that condition $(i)$ cannot be satisfied and some node rejects.     
  
  We now analyse the proof-size of the protocol:  the certification size for the number of nodes $n(G)$ and the paths constructed is $\cO(\log n)$, given by  \Cref{prop:numbernode} and \Cref{prop:stpath}. On the other hand, for each $v\in V$, the values $b_1(v)$, $b_2(v)$, $t_1(v)$, $t_2(v)$, $p_v$, $q_v$ are computable in $\cO(\log n)$ bits as all values are in $[2n]$. Overall the total communication is $\cO(\log n)$. \\
  
  Now, we can use the above protocol to recognize \permutation, the class of permutation graphs. A graph is said to be a permutation graph if there exists a collection of points $\{\ell_1(v)\}_{v\in V}$ and $\{\ell_2(v)\}_{v\in V}$ inscribed in two parallel lines such that $\{u,v\}\in E$ iff $(\ell_1(v)-\ell_1(u))~(\ell_2(v)-\ell_2(u))~<~0$. This means that the line with extremes $\ell_1(v)$ and $\ell_2(v)$ must cross the line with extremes $\ell_1(u)$ and $\ell_2(u)$, as in Figure~\ref{fig:Expermutation}.
  
  By the same argument of the trapezoid model, we can enumerate the points in both lines, so without loss of generality we can assume the collections $\{\ell_1(v)\}_{v\in V}$ and $\{\ell_2(v)\}_{v\in V}$ are permutations from $V$ to $[n]$. We say that such collection $\{\ell_1(v), \ell_2(v)\}_{v\in V}$ is a proper permutation model of $G$.
    
Let $G=(V,E)$ be a graph. We define a trapezoid model ${T_v}_{v\in V}$ of $G$ to be {\it consecutive} if for every $v\in V$, we have $b_2(v) = b_1(v)+1$ and $t_2(v) = t_1(v)+1$. We denote the class of graphs that have a consecutive trapezoid model by {\sc ConTrapezoid}.
  
  \begin{proposition}
  	\label{prop:consecutive}
  {\sc ConTrapezoid = \permutation}.
  \end{proposition}
\begin{proof}
	Let $G\in$ {\sc ConTrapezoid} and $T_v$ its consecutive trapezoid model. By definition, for all $v\in V$, exists $i_v,j_v\in \Z_n$ such that $t_1(v) =2i_v+1$, $b_1(v) = 2j_v+1$ and no pair of nodes have the same values $i_v,j_v$ defining its vertices $t_1$ and $b_1$. Then, we can define the functions $\ell_1$ and $\ell_2$ such that $\ell_1(v) = i_v$ and $\ell_2v) = j_v$, respectively. This clearly represents a valid permutation model of $G$ given that $\{T_v\}_{v\in V}$ its a valid consecutive trapezoid model.
	
	In the same way, if $G$ it is a permutation graph and $\{\ell_1(v),\ell_2(v)\}_{v\in V}$ its permutation model, then if we define $\{T_v\}_{v\in V}$ such that $t_1(v) = 2\ell_2(v)-1 $, $t_2(v) = 2\ell_2(v)$, $b_1(v) = 2\ell_1(v)-1$ and $b_2(v) = 2\ell_1(v)$, this is a valid consecutive trapezoid model of $G$.
\end{proof}

 Using \Cref{prop:consecutive}, it is straightforward to modify the protocol described in \Cref{proto:trapezoid} by having the prover send consecutive vertices and allowing the nodes to check them accordingly to recognize the class of permutation graphs.
  
  \begin{corollary}
  There is a PLS for $\permutation$ with proof-size of $\mathcal{O}(\log n)$ bits.
  \end{corollary}
  
\section{Lower Bounds}\label{sec:lower}

In this section, logarithmic lower bounds are given in the certificate sizes of any PLS that recognizes the class of interval, circular-arc, chordal, permutation, and trapezoid graphs. In order to do so, two techniques are used, each one explained in \Cref{subsec:intervallower} and \Cref{subsec:trapelower}, respectively.

  \subsection{Interval, Chordal and Circular Arc Graphs}
  \label{subsec:intervallower}
  
  To prove the lower bound for interval, circular-arc and chordal graphs, we adapt a construction by Göös and Soumela for locally checkable proofs \cite{goos2016locally} where the main idea is to construct a collection of yes-instances of each class and a no-instance that will be indistinguishable from the yes-instances if we assume there exists a PLS with proof-size of $o(\log n)$ bits to recognize each class, giving a contradiction.

Before giving our lower bounds, we need to define a combinatorial result. A hyper-graph is a generalization of a graph, where each hyper-edge is a subset of nodes. An $r$-uniform hyper-graph is a hyper-graph where each hyper-edge has the same cardinality $r$. A graph $G$ is simply a 2-uniform hyper-graph. 
Let $K^{(r)(\ell)}$ be an $r$-uniform hyper-graph, where we can split the set of nodes into $r$ parts, with size $\ell$ each, such that we have a hyper-edge for any selection of elements from each of the $r$ different parts, 
In 1964, Erd{\"o}s showed the following result.

\begin{proposition} \cite{erdos1964extremal}
	\label{lem:erdos}
	Let $G$ be an $r$-uniform hyper-graph. If $G$ does not contain a $K^{(r)(\ell)}$ as a subgraph, then $|E(G)|\leq n^{1-1/\ell^{r-1}}$.
\end{proposition}

We are now ready to show our result.

  \begin{theorem}
  	Any PLS that recognizes \interval, \chordal{ } and \circularArc{ } needs a proof-size of $\Omega(\log n)$ bits.
  \end{theorem}

\begin{proof}

  Consider $n$ to be even and  $A$ to be a partition of $[1,n^2]$ into $n$ sets of size $n$, and $B$ as a partition of $[n^2+1 , 2n^2]$ in a similar manner. Let $\mathcal{G}$ be a family of $n$-node graphs satisfying the property $\mathrm{P}$ of belonging to the classes above mentioned. 
  
  Set $\mathcal{G}_A $ and $\mathcal{G}_B$ to be the set of labeled graphs in $\mathcal{G}$, with label sets picked from $A$ and $B$ respectively. Let $F_a$ be a graph in $\mathcal{G}_A$ and $F_b$ a graph from $\mathcal{G}_B$. Finally consider two disjoint sets $C$ and $D$  in $[2n^2+1, 3n^2]$ of size $n$.
  
  For $(F_a,F_b,c,d) \in \mathcal{G}_A\times \mathcal{G}_B\times C\times D$ let $G(F_a,F_b,c,d)$ be the graph defined by the disjoint union of graphs  $F_a$ and $F_b$ plus four additional nodes $y_A,  y_B, c, d$. The nodes $y_A, y_B$ are labelled with different numbers in a set from $ [2n^2 +1, 3n^2]$ disjoint from $C$ and $D$. The nodes $y_A$ , $y_B$ , $c$ and $d$ form a clique $K_4$, while the nodes $y_A$ is connected to some node $v_a$ in $F_a$. Node $y_B$ is similarly adjacent to some node $v_b$ in $F_b$.  Observe that all nodes in $F_a$ communicate with $F_b$ only through the nodes $ y_A, y_B, c$ and $d$.
  
  \begin{figure}[!ht]
  	\begin{center}
  		\includegraphics[width=0.4\linewidth]{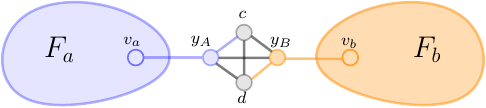}
  		\caption{ A yes-instance for \interval\ and its super-classes, as $F_a$ and $F_b$ are two interval graphs which are connected through a 4-clique in the middle}
  		\label{fig:lowerboundclique}
  	\end{center}
  	
  \end{figure}
  
  Let $\mathcal{P}$ be a PLS verifying the property $\mathrm{P}$ with bandwidth $K = \delta \log \log n$ and error probability $\eps$, for some $\delta,\epsilon >0$.
 
Let ${y_A,y_B, c, d}$ be the \emph{bridge} of $G(F_a, F_b, c, d)$. Without loss of generality, we assume that protocol $\mathcal{P}$ satisfies the condition that the nodes in the bridge ${y_A,y_B, c, d}$ receive the same proof. If the protocol does not satisfy this condition, we can construct a new protocol with proof length $4L$, where $L$ denotes the original message size, by having each node select their respective portion of the proof and then follow the original protocol.
  
 We define $M=\{m_v\}_{v\in V}$ as the set of certificates indexed by vertices $v\in V(G[F_a,F_b,c,d])$, where $m_v$ denotes the certificate that the prover sends to node $v$ in protocol $\mathcal{P}$. Let $\mathcal{M} \subseteq \{0,1\}^K$ be the set of certificates such that if it is assigned to the nodes ${y_A,y_B,x_A,x_B}$, it can be extended to a proof assignment for the nodes in both $F{a}$ and $F_{b}$, causing them to accept whenever the bridge accepts.

  Now consider the complete 4-partite, 4-uniform hyper-graph graph $\tilde{G} = A \cup B\cup C\cup D$. For each $a\in A, b\in B, c\in C$ and $d\in D$, color the edge $\{a,b,c,d\}$ with the certificate $m_{abcd}\in K$ of a possible assignment to the nodes of the bridge. There are at most $2^{K}$ possible certificates $m_{abcd}$ and $n^4$ hyper-edges in $\tilde{G}$. Therefore, by the pigeonhole principle, there exists a monochromatic set of hyper-edges $H$ of size at least $\frac{n^4}{2^{K}}$. 
  
  \begin{figure}[!ht]
  	\centering
  	\includegraphics[width=0.4\linewidth]{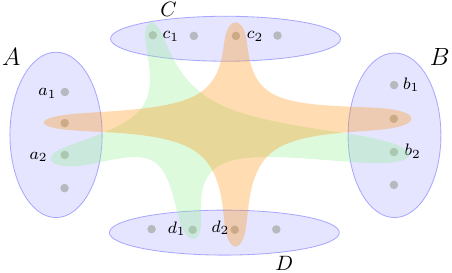}
  	\caption{Auxiliary (complete) 4-uniform, 4 partite hyper-graph $\tilde{G} = A\cup B\cup C\cup D$ with each node in $A$ and $B$ representing a subgraph $F_a, F_b$ and nodes in $C$ and $D$ representing single nodes in the original graph construction. The green and orange edge blocks are a pair of blocks from a monochromatic $K^{(4)}(2)$ structure present in the graph}
  	\label{fig:hypergoos}
  \end{figure}

  Observe that  for sufficiently small $\delta$ and large $n$, $2^{K} = (\log n)^{\delta k \log^{\delta k} (n)} = o(n^{1/{8}})$. Indeed, if $n > 2^{k\delta}$ and $\delta < 1/(2^{4}k)$  have that $ \delta k\log^{\delta k}(n) \leq \log^{2\delta k}(n) <\frac{1}{8} \log n$.  Now, following a result from Ërdos, described in Lemma~ \ref{lem:erdos}, by setting $\ell =2, r=4$ we have that there exists a $K^{(r)}(\ell)$ subgraph in $\tilde{G}$ induced by $H$. That is, the complete $r$-uniform, $r$-partite hyper-graph, where each part has a size exactly $\ell$. Let $\{a_i,b_i, c_i, d_i\}_{i=1}^2$ be the nodes involved in such a graph.

  Consider now the graph $G(a_i, b_i, c_i, d_i)$ defined as follows: First. take a disjoint union of $F_{a_1}, F_{b_1}, F_{a_2}$ and $F_{b_2}$. Then, for each $i \in \{1, 2\}$ add nodes $y_A^{i}, y_B^{i}, c^{i}, d^{i}$, labelled with different labels in $[2n^2+1, 3n^2]$  correspondent to the yes instances formed by the graphs $F_{a_i}$  $F_{b_j}$, and the nodes $c_k$ and $d_h$ . 	For each $i\in \{1,2\}$, the node $y_A^{i}$  is adjacent to $y_B^i$, $c_i$, $d_{i+1}$ and the node $v_{a_i}$ of $F_{a_i}$. Also, the node $y_B^i$ is adjacent to $c^i$, $d_{i}$ and the node $v_{b_i}$ of $F_{b_i}$ and $c_i$ is adjacent to $d_{i+1}$ Where the $i+1$ is taken$\mod 2$.

 We must demonstrate that the graph $G(a_1, b_1, a_2, b_2)$ is a No-instance for the property $\mathrm{P}$ of belonging to the classes \propInterval, \interval, \propCircular, \circularArc{ } and \chordal.

  \begin{figure}[!ht]
  	\centering
  	\includegraphics[width=0.4\linewidth]{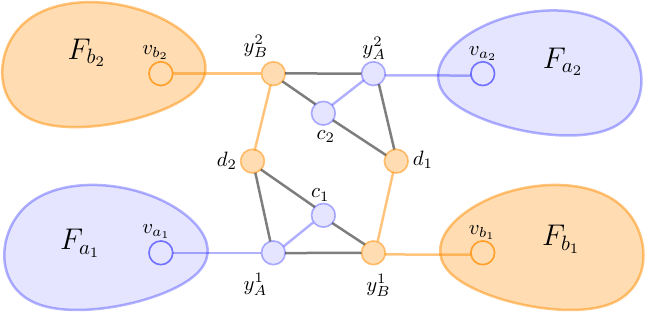}
  	\caption{ A No-instance for \interval\ and its super-classes, as the graph admits a large cycle without any chords (and with large diameter subgraphs $F_{a_i}, F_{b_j}$), it can not admit a representation through intervals or circular arcs, nor can be a chordal graph.}
  	\label{fig:fakehyperbound}
  \end{figure}

  \begin{itemize}
  	\item As for the classes $\interval$ and $\propInterval$ we simply define $F_a$ and $F_b$ to be a pair of proper interval graphs of size $\cO(n)$, then $G(F_a, F_b, c, d)$ also admits a representation through proper intervals as we simply connect both graphs through their extremes by a small clique. Finally, the newly constructed graph has an induced 6-cycle and therefore can not have a representation by intervals.
  	\item As for $\propCircular$ and $\circularArc$ by using the same construction (by assuring that each part has a large diameter) we also have a valid instance as (proper) interval graphs are in particular (proper) circular arc graphs. We simply consider their representation through intervals in the real line as a big arc in a portion of the circle.  As for the newly obtained instance, we have a large, induced cycle which is consistent with a construction for a circular arc graph, while we also have large paths on each side that are in conflict with the cycle as a (proper) circular arc graph behaves locally like an interval graph and therefore can not have an asteroidal triple (three nodes at the extreme of each interval graph).
  	\item Finally, for $\chordal$ we can use the same graph described for the class \interval as they are also chordal graphs. We have that the newly constructed graph has a 6-cycle without any chords. Therefore, it is not in \chordal.
  \end{itemize}

This means that if we run protocol $\mathcal{P}$ in instance $G(a_1,b_i,c_i,d_i)$ for $i\in\{1,2\}$, at least one node should reject. But the local information given to each node in $G(a_1,b_i,c_i,d_i)$ it is the same that in a yes-instance $G(F_{a_i},F_{b_i},c_id_i)$, as they have the same neighbors with the same \textit{id}'s and labels, and the nodes of the bridge receives the same certificate $m_{a_ib_ic_id_i}$ that makes the yes-instance accepts, so protocol $\mathcal{P}$ makes all the nodes accept in instance  $G(a_1,b_i,c_i,d_i)$, that is not part of the class, making a contradiction.

\end{proof}

\subsection{Trapezoid and Permutation Graphs}
\label{subsec:trapelower}

For the class of trapezoid and permutation graphs, we use a technique given by Fraigniaud et al \cite{fraigniaud2019randomized}, called \emph{crossing edge}, and which we detail as follows. Let $G=(V,E)$ be a graph and let $H_1 = (V_1,E_1)$ and $H_2 = (V_2,E_2)$ be two subgraphs of $G$. We say that $H_1$ and $H_2$ are independent if and only if $V_1\cap V_2 = \emptyset$ and $E\cap (V_1\times V_2) = \emptyset$.

To prove a lower bound on the proof-size of any PLS that recognizes the remaining classes, we use the following results of Fraigniaud et al. \cite{fraigniaud2019randomized}.

\begin{definition}[\cite{fraigniaud2019randomized}]
	\label{teo:fraignetwork}
	Let $G=(V,E)$ be a graph and let $H_1 = (V_1,E_1)$ and $H_2 = (V_2,E_2)$ be two independent isomorphic subgraphs of $G$ with isomorphism $\sigma\colon V_1\to V_2$. The {\sc crossing} of $G$ induced by $\sigma$, denoted by $\sigma_{\bowtie}(G)$, is the graph obtained from $G$ by replacing every pair of edges $\{u,v\}\in E_1$ and $\{\sigma(u),\sigma(v)\}\in E_2$, by the pair $\{u,\sigma(v)\}$ and $\{\sigma(u),v\}$.
\end{definition}

Then, a lower bound to any PLS is stated as follows.

\begin{theorem}[\cite{fraigniaud2019randomized}] 
	\label{teo:lowerbound}
	Let $\mathcal{F}$ be a family of network configurations, and let $\mathcal{P}$ be a boolean predicate over $\mathcal{F}$. Suppose that there is a configuration $G_s\in\mathcal{F}$ satisfying that (1) $G$ contains as subgraphs $r$ pairwise independent isomorphic copies $H_1,...,H_r$ with $s$ edges each, and (2) there exists  $r$ port-preserving isomorphisms $\sigma_i\colon V(H_1)\to V(H_i)$ such that for every $i\neq j$, the isomorphism $\sigma^{ij} = \sigma_i\circ\sigma_j^{-1} $ satisfies $\mathcal{P}(G_s)\neq \mathcal{P}(\sigma^{ij}_{\bowtie}(G)_s)$. Then,  the verification complexity of any proof-labeling scheme for $\mathcal{P}$ and $\mathcal{F}$ is $\Omega\left(\dfrac{log(r)}{s}\right)$.
\end{theorem}

We prove the lower bounds remaining by constructing a specific graph family with respective isomorphisms satisfying \Cref{teo:fraignetwork} and hypothesis of \Cref{teo:lowerbound} to conclude.

\begin{theorem}
	Any PLS for \permutation\ or \trapezoid\ needs a proof-size of $\Omega\left(\log n \right)$ bits.
\end{theorem}

\begin{proof}

 First, let $\mathcal{F} = \{(Q_n,\id)\}$ be a collection of network configurations, where each graph $Q_n$ consists of $5n$ nodes forming a path $\{v_1, \dots, v_{5n}\}$ where we add the edge $\{v_{5i-3}, v_{5i-1}\}$, for each $i \in [n]$. It is easy to see that for each $n>0$, $Q_n$ is a permutation graph (and then also a trapezoid graph), and therefore $\mathcal{F}\subseteq \permutation$ and $\mathcal{F}\subseteq \trapezoid$.  In \cref{fig:lowerboundperm1} is depicted the graph $Q_3$ and its corresponding permutation model. 

\begin{figure}[h!]
	\centering
	\input{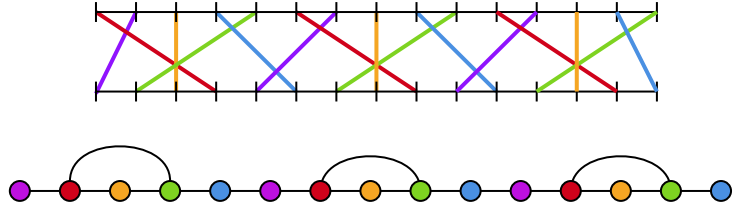}
	\caption{Graph $Q_3$ and a permutation model for $Q_3$. }
	\label{fig:lowerboundperm1}
\end{figure}

Given $Q_n$, consider the subgraphs $H_i = \{v_{5i-2}, v_{5i-1}\}$, for each $i \in [n]$, and the isomorphism $\sigma_i\colon V(H_1)\to V(H_i)$ such that $\sigma_i(v_3) = v_{5i-2}$ and $\sigma_i(v_4) = v_{5i-1}$.

\begin{lemma}
	\label{lemma:notrap}
	For each $i\neq j$, the graph $\sigma^{ij}_{\bowtie}(Q_n)$ it is neither a permutation graph nor a trapezoid graph with $\sigma_i\colon V(H_1)\to V(H_i)$ such that $\sigma_i(v_3) = v_{5i-2}$ and $\sigma_i(v_4) = v_{5i-1}$.  
\end{lemma}
  
\begin{proof}
	Given $i< j$,  by definition of $\sigma^{ij}\colon V(H_j)\to V(H_i)$ in $\sigma^{ij}_{\bowtie}(Q_n)$ the nodes $v_{5j-3}$, $v_{5j-2}$, $v_{5i-1}$, $v_{5i-3}$, $v_{5i-2}$, $v_{5j-1}$ form an induced cycle of length $6$ (see \cref{fig:lowerboundperm2} for an example). 
	
	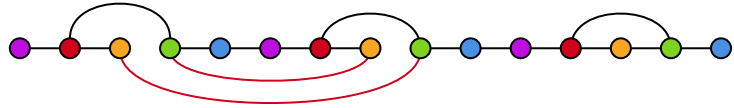
\begin{figure}[h!]
		\centering
		\tikzset{every picture/.style={line width=0.75pt}} 

\begin{tikzpicture}[x=0.75pt,y=0.75pt,yscale=-1,xscale=1]

\draw    (190,77.46) -- (240,77.46) ;
\draw    (265,77.46) -- (365,77.46) ;
\draw    (390,77.46) -- (540,77.46) ;
\draw [color={rgb, 255:red, 208; green, 2; blue, 27 }  ,draw opacity=1 ]   (240,78.54) .. controls (240,113.82) and (390,113.82) .. (390,78.54) ;
\draw [color={rgb, 255:red, 208; green, 2; blue, 27 }  ,draw opacity=1 ]   (265,78.54) .. controls (265.11,98.71) and (365.11,99.04) .. (365,78.54) ;
\draw    (215,72.46) .. controls (215.11,48.85) and (265.11,49.51) .. (265,72.46) ;
\draw    (340,77.46) .. controls (340.11,53.85) and (390.11,54.51) .. (390,77.46) ;
\draw  [fill={rgb, 255:red, 245; green, 166; blue, 35 }  ,fill opacity=1 ] (235,77.46) .. controls (235,74.7) and (237.24,72.46) .. (240,72.46) .. controls (242.76,72.46) and (245,74.7) .. (245,77.46) .. controls (245,80.22) and (242.76,82.46) .. (240,82.46) .. controls (237.24,82.46) and (235,80.22) .. (235,77.46) -- cycle ;
\draw  [fill={rgb, 255:red, 126; green, 211; blue, 33 }  ,fill opacity=1 ] (260,77.46) .. controls (260,74.7) and (262.24,72.46) .. (265,72.46) .. controls (267.76,72.46) and (270,74.7) .. (270,77.46) .. controls (270,80.22) and (267.76,82.46) .. (265,82.46) .. controls (262.24,82.46) and (260,80.22) .. (260,77.46) -- cycle ;
\draw  [fill={rgb, 255:red, 208; green, 2; blue, 27 }  ,fill opacity=1 ] (210,77.46) .. controls (210,74.7) and (212.24,72.46) .. (215,72.46) .. controls (217.76,72.46) and (220,74.7) .. (220,77.46) .. controls (220,80.22) and (217.76,82.46) .. (215,82.46) .. controls (212.24,82.46) and (210,80.22) .. (210,77.46) -- cycle ;
\draw  [fill={rgb, 255:red, 74; green, 144; blue, 226 }  ,fill opacity=1 ] (285,77.46) .. controls (285,74.7) and (287.24,72.46) .. (290,72.46) .. controls (292.76,72.46) and (295,74.7) .. (295,77.46) .. controls (295,80.22) and (292.76,82.46) .. (290,82.46) .. controls (287.24,82.46) and (285,80.22) .. (285,77.46) -- cycle ;
\draw  [fill={rgb, 255:red, 189; green, 16; blue, 224 }  ,fill opacity=1 ] (310,77.46) .. controls (310,74.7) and (312.24,72.46) .. (315,72.46) .. controls (317.76,72.46) and (320,74.7) .. (320,77.46) .. controls (320,80.22) and (317.76,82.46) .. (315,82.46) .. controls (312.24,82.46) and (310,80.22) .. (310,77.46) -- cycle ;
\draw  [fill={rgb, 255:red, 208; green, 2; blue, 27 }  ,fill opacity=1 ] (335,77.46) .. controls (335,74.7) and (337.24,72.46) .. (340,72.46) .. controls (342.76,72.46) and (345,74.7) .. (345,77.46) .. controls (345,80.22) and (342.76,82.46) .. (340,82.46) .. controls (337.24,82.46) and (335,80.22) .. (335,77.46) -- cycle ;
\draw  [fill={rgb, 255:red, 189; green, 16; blue, 224 }  ,fill opacity=1 ] (185,77.46) .. controls (185,74.7) and (187.24,72.46) .. (190,72.46) .. controls (192.76,72.46) and (195,74.7) .. (195,77.46) .. controls (195,80.22) and (192.76,82.46) .. (190,82.46) .. controls (187.24,82.46) and (185,80.22) .. (185,77.46) -- cycle ;
\draw  [fill={rgb, 255:red, 245; green, 166; blue, 35 }  ,fill opacity=1 ] (360,77.46) .. controls (360,74.7) and (362.24,72.46) .. (365,72.46) .. controls (367.76,72.46) and (370,74.7) .. (370,77.46) .. controls (370,80.22) and (367.76,82.46) .. (365,82.46) .. controls (362.24,82.46) and (360,80.22) .. (360,77.46) -- cycle ;
\draw  [fill={rgb, 255:red, 126; green, 211; blue, 33 }  ,fill opacity=1 ] (385,77.46) .. controls (385,74.7) and (387.24,72.46) .. (390,72.46) .. controls (392.76,72.46) and (395,74.7) .. (395,77.46) .. controls (395,80.22) and (392.76,82.46) .. (390,82.46) .. controls (387.24,82.46) and (385,80.22) .. (385,77.46) -- cycle ;
\draw  [fill={rgb, 255:red, 74; green, 144; blue, 226 }  ,fill opacity=1 ] (410,77.46) .. controls (410,74.7) and (412.24,72.46) .. (415,72.46) .. controls (417.76,72.46) and (420,74.7) .. (420,77.46) .. controls (420,80.22) and (417.76,82.46) .. (415,82.46) .. controls (412.24,82.46) and (410,80.22) .. (410,77.46) -- cycle ;
\draw    (465,77.46) .. controls (465.11,53.85) and (515.11,54.51) .. (515,77.46) ;
\draw  [fill={rgb, 255:red, 189; green, 16; blue, 224 }  ,fill opacity=1 ] (435,77.46) .. controls (435,74.7) and (437.24,72.46) .. (440,72.46) .. controls (442.76,72.46) and (445,74.7) .. (445,77.46) .. controls (445,80.22) and (442.76,82.46) .. (440,82.46) .. controls (437.24,82.46) and (435,80.22) .. (435,77.46) -- cycle ;
\draw  [fill={rgb, 255:red, 208; green, 2; blue, 27 }  ,fill opacity=1 ] (460,77.46) .. controls (460,74.7) and (462.24,72.46) .. (465,72.46) .. controls (467.76,72.46) and (470,74.7) .. (470,77.46) .. controls (470,80.22) and (467.76,82.46) .. (465,82.46) .. controls (462.24,82.46) and (460,80.22) .. (460,77.46) -- cycle ;
\draw  [fill={rgb, 255:red, 245; green, 166; blue, 35 }  ,fill opacity=1 ] (485,77.46) .. controls (485,74.7) and (487.24,72.46) .. (490,72.46) .. controls (492.76,72.46) and (495,74.7) .. (495,77.46) .. controls (495,80.22) and (492.76,82.46) .. (490,82.46) .. controls (487.24,82.46) and (485,80.22) .. (485,77.46) -- cycle ;
\draw  [fill={rgb, 255:red, 126; green, 211; blue, 33 }  ,fill opacity=1 ] (510,77.46) .. controls (510,74.7) and (512.24,72.46) .. (515,72.46) .. controls (517.76,72.46) and (520,74.7) .. (520,77.46) .. controls (520,80.22) and (517.76,82.46) .. (515,82.46) .. controls (512.24,82.46) and (510,80.22) .. (510,77.46) -- cycle ;
\draw  [fill={rgb, 255:red, 74; green, 144; blue, 226 }  ,fill opacity=1 ] (535,77.46) .. controls (535,74.7) and (537.24,72.46) .. (540,72.46) .. controls (542.76,72.46) and (545,74.7) .. (545,77.46) .. controls (545,80.22) and (542.76,82.46) .. (540,82.46) .. controls (537.24,82.46) and (535,80.22) .. (535,77.46) -- cycle ;

\end{tikzpicture}
		\caption{Graph $\sigma^{12}_{\bowtie}(Q_3)$, where in red are represented the crossing edges.}
		\label{fig:lowerboundperm2}
	\end{figure}
	
	As a trapezoid graph have induced cycles of length at most $4$, we deduce that $\sigma^{ij}_{\bowtie}(Q_n)$ is not a trapezoid graph.
	
	Finally, as the class of permutation graph is contained in the class of trapezoid graphs, then $Q_n$ neither it is a permutation graph.
\end{proof}

Then, we have that for all $n>0$, the graph $Q_n$ is a permutation and a trapezoid graph, but $\sigma^{ij}_{\bowtie}(Q_n)$ is neither permutation nor trapezoid graph. Since there are $r=\mathcal{O}(n)$ isomorphisms $\sigma_{i}$ and subgraphs $H_i$, each with one edge, it follows from \Cref{teo:lowerbound} that any PLS that recognizes \trapezoid\ and \permutation\ needs a proof size of $\Omega(\log n)$ bits.

\end{proof}


\bibliography{sn-article}

\begin{thebibliography}{10}

\bibitem{abouelhoda2005chaining}
{\sc M.~I. Abouelhoda and E.~Ohlebusch}, {\em Chaining algorithms for multiple
  genome comparison}, Journal of Discrete Algorithms, 3 (2005), pp.~321--341.

\bibitem{asdre2007harmonious}
{\sc K.~Asdre, K.~Ioannidou, and S.~D. Nikolopoulos}, {\em The harmonious
  coloring problem is np-complete for interval and permutation graphs},
  Discrete Applied Mathematics, 155 (2007), pp.~2377--2382.

\bibitem{BalliuDFO18}
{\sc A.~Balliu, G.~D'Angelo, P.~Fraigniaud, and D.~Olivetti}, {\em What can be
  verified locally?}, J. Comput. Syst. Sci., 97 (2018), pp.~106--120.

\bibitem{beeri1983desirability}
{\sc C.~Beeri, R.~Fagin, D.~Maier, and M.~Yannakakis}, {\em On the desirability
  of acyclic database schemes}, Journal of the ACM (JACM), 30 (1983),
  pp.~479--513.

\bibitem{BickKO22}
{\sc A.~Bick, G.~Kol, and R.~Oshman}, {\em Distributed zero-knowledge proofs
  over networks}, in 33rd ACM-SIAM Symposium on Discrete Algorithms (SODA),
  2022, pp.~2426--2458.

\bibitem{blair1993introduction}
{\sc J.~R. Blair and B.~Peyton}, {\em An introduction to chordal graphs and
  clique trees}, in Graph theory and sparse matrix computation, Springer, 1993,
  pp.~1--29.

\bibitem{bodlaender1998partial}
{\sc H.~L. Bodlaender}, {\em A partial k-arboretum of graphs with bounded
  treewidth}, Theoretical computer science, 209 (1998), pp.~1--45.

\bibitem{bodlaender1992two}
{\sc H.~L. Bodlaender, M.~R. Fellows, and T.~J. Warnow}, {\em Two strikes
  against perfect phylogeny}, in International Colloquium on Automata,
  Languages, and Programming, Springer, 1992, pp.~273--283.

\bibitem{bousquet2021distributed}
{\sc N.~Bousquet, L.~Feuilloley, M.~Heinrich, and M.~Rabie}, {\em Distributed
  recoloring of interval and chordal graphs}, arXiv preprint arXiv:2109.06021,
  (2021).

\bibitem{BousquetFP21}
{\sc N.~Bousquet, L.~Feuilloley, and T.~Pierron}, {\em Local certification of
  graph decompositions and applications to minor-free classes}, in 25th
  International Conference on Principles of Distributed Systems (OPODIS),
  vol.~217 of LIPIcs, Schloss Dagstuhl - Leibniz-Zentrum f{\"{u}}r Informatik,
  2021, pp.~22:1--22:17.

\bibitem{bousquet2021local}
{\sc N.~Bousquet, L.~Feuilloley, and T.~Pierron}, {\em Local certification of
  mso properties for bounded treedepth graphs}, arXiv preprint
  arXiv:2110.01936,  (2021).

\bibitem{brandstadt1999graph}
{\sc A.~Brandstadt, J.~P. Spinrad, et~al.}, {\em Graph classes: a survey},
  vol.~3, Siam, 1999.

\bibitem{CrescenziFP19}
{\sc P.~Crescenzi, P.~Fraigniaud, and A.~Paz}, {\em Trade-offs in distributed
  interactive proofs}, in 33rd International Symposium on Distributed Computing
  (DISC), vol.~146 of LIPIcs, Schloss Dagstuhl - Leibniz-Zentrum f{\"{u}}r
  Informatik, 2019, pp.~13:1--13:17.

\bibitem{dagan1988trapezoid}
{\sc I.~Dagan, M.~C. Golumbic, and R.~Y. Pinter}, {\em Trapezoid graphs and
  their coloring}, Discrete Applied Mathematics, 21 (1988), pp.~35--46.

\bibitem{erdos1964extremal}
{\sc P.~Erd{\"o}s}, {\em On extremal problems of graphs and generalized
  graphs}, Israel Journal of Mathematics, 2 (1964), pp.~183--190.

\bibitem{EsperetL22}
{\sc L.~Esperet and B.~L{\'{e}}v{\^{e}}que}, {\em Local certification of graphs
  on surfaces}, Theor. Comput. Sci., 909 (2022), pp.~68--75.

\bibitem{feuilloley2021introduction}
{\sc L.~Feuilloley}, {\em Introduction to local certification}, Discrete
  Mathematics \& Theoretical Computer Science, 23 (2021).

\bibitem{FeuilloleyBP22}
{\sc L.~Feuilloley, N.~Bousquet, and T.~Pierron}, {\em What can be certified
  compactly? compact local certification of {MSO} properties in tree-like
  graphs}, in 41st ACM Symposium on Principles of Distributed Computing (PODC),
  2022, pp.~131--140.

\bibitem{FeuilloleyFH21}
{\sc L.~Feuilloley, P.~Fraigniaud, and J.~Hirvonen}, {\em A hierarchy of local
  decision}, Theor. Comput. Sci., 856 (2021), pp.~51--67.

\bibitem{FeuilloleyFHPP21}
{\sc L.~Feuilloley, P.~Fraigniaud, J.~Hirvonen, A.~Paz, and M.~Perry}, {\em
  Redundancy in distributed proofs}, Distributed Computing, 34 (2021),
  pp.~113--132.

\bibitem{feuilloley2020compact}
{\sc L.~Feuilloley, P.~Fraigniaud, P.~Montealegre, I.~Rapaport,
  {\'E}.~R{\'e}mila, and I.~Todinca}, {\em Compact distributed certification of
  planar graphs}, Algorithmica,  (2021), pp.~1--30.

\bibitem{FeuilloleyH18}
{\sc L.~Feuilloley and J.~Hirvonen}, {\em Local verification of global proofs},
  in 32nd International Symposium on Distributed Computing, vol.~121 of LIPIcs,
  Schloss Dagstuhl - Leibniz-Zentrum f{\"{u}}r Informatik, 2018,
  pp.~25:1--25:17.

\bibitem{FraigniaudGNP21}
{\sc P.~Fraigniaud, F.~L. Gall, H.~Nishimura, and A.~Paz}, {\em Distributed
  quantum proofs for replicated data}, in 12th Innovations in Theoretical
  Computer Science Conference (ITCS), vol.~185 of LIPIcs, Schloss Dagstuhl -
  Leibniz-Zentrum f{\"{u}}r Informatik, 2021, pp.~28:1--28:20.

\bibitem{FraigniaudMRT22}
{\sc P.~Fraigniaud, P.~Montealegre, I.~Rapaport, and I.~Todinca}, {\em A
  meta-theorem for distributed certification}, in 29th International Colloquium
  on Structural Information and Communication Complexity (SIROCCO), vol.~13298
  of LNCS, Springer, 2022, pp.~116--134.

\bibitem{fraigniaud2019randomized}
{\sc P.~Fraigniaud, B.~Patt-Shamir, and M.~Perry}, {\em Randomized
  proof-labeling schemes}, Distributed Computing, 32 (2019), pp.~217--234.

\bibitem{gardi2007roberts}
{\sc F.~Gardi}, {\em The roberts characterization of proper and unit interval
  graphs}, Discrete Mathematics, 307 (2007), pp.~2906--2908.

\bibitem{garey1980complexity}
{\sc M.~R. Garey, D.~S. Johnson, G.~L. Miller, and C.~H. Papadimitriou}, {\em
  The complexity of coloring circular arcs and chords}, SIAM Journal on
  Algebraic Discrete Methods, 1 (1980), pp.~216--227.

\bibitem{golumbic2004algorithmic}
{\sc M.~C. Golumbic}, {\em Algorithmic graph theory and perfect graphs},
  Elsevier, 2004.

\bibitem{goos2016locally}
{\sc M.~G{\"o}{\"o}s and J.~Suomela}, {\em Locally checkable proofs in
  distributed computing}, Theory of Computing, 12 (2016), pp.~1--33.

\bibitem{halldorsson2020improved}
{\sc M.~M. Halld{\'o}rsson and C.~Konrad}, {\em Improved distributed algorithms
  for coloring interval graphs with application to multicoloring trees},
  Theoretical Computer Science, 811 (2020), pp.~29--41.

\bibitem{hoang1994efficient}
{\sc C.~T. Ho{\`a}ng}, {\em Efficient algorithms for minimum weighted colouring
  of some classes of perfect graphs}, Discrete Applied Mathematics, 55 (1994),
  pp.~133--143.

\bibitem{kante2013enumeration}
{\sc M.~M. Kant{\'e}, V.~Limouzy, A.~Mary, L.~Nourine, and T.~Uno}, {\em On the
  enumeration and counting of minimal dominating sets in interval and
  permutation graphs}, in International Symposium on Algorithms and
  Computation, Springer, 2013, pp.~339--349.

\bibitem{kaplan2006simpler}
{\sc H.~Kaplan and Y.~Nussbaum}, {\em A simpler linear-time recognition of
  circular-arc graphs}, in Scandinavian Workshop on Algorithm Theory, Springer,
  2006, pp.~41--52.

\bibitem{kaplan1996pathwidth}
{\sc H.~Kaplan and R.~Shamir}, {\em Pathwidth, bandwidth, and completion
  problems to proper interval graphs with small cliques}, SIAM Journal on
  Computing, 25 (1996), pp.~540--561.

\bibitem{kennedy2006strictly}
{\sc W.~Kennedy, G.~Lin, and G.~Yan}, {\em Strictly chordal graphs are leaf
  powers}, Journal of Discrete Algorithms, 4 (2006), pp.~511--525.

\bibitem{kol2018interactive}
{\sc G.~Kol, R.~Oshman, and R.~R. Saxena}, {\em Interactive distributed
  proofs}, in ACM Symposium on Principles of Distributed Computing, ACM, 2018,
  pp.~255--264.

\bibitem{konrad2019distributed}
{\sc C.~Konrad and V.~Zamaraev}, {\em Distributed minimum vertex coloring and
  maximum independent set in chordal graphs}, in 44th International Symposium
  on Mathematical Foundations of Computer Science (MFCS 2019), Schloss
  Dagstuhl-Leibniz-Zentrum fuer Informatik, 2019.

\bibitem{KormanKP10}
{\sc A.~Korman, S.~Kutten, and D.~Peleg}, {\em Proof labeling schemes},
  Distributed Comput., 22 (2010), pp.~215--233.

\bibitem{korman2010proof}
{\sc A.~Korman, S.~Kutten, and D.~Peleg}, {\em Proof labeling schemes},
  Distributed Computing, 22 (2010), pp.~215--233.

\bibitem{kratsch2006certifying}
{\sc D.~Kratsch, R.~M. McConnell, K.~Mehlhorn, and J.~P. Spinrad}, {\em
  Certifying algorithms for recognizing interval graphs and permutation
  graphs}, SIAM Journal on Computing, 36 (2006), pp.~326--353.

\bibitem{lin2000phylogenetic}
{\sc G.-H. Lin, P.~E. Kearney, and T.~Jiang}, {\em Phylogenetic k-root and
  steiner k-root}, in International Symposium on Algorithms and Computation,
  Springer, 2000, pp.~539--551.

\bibitem{ma19942}
{\sc T.-H. Ma and J.~P. Spinrad}, {\em On the 2-chain subgraph cover and
  related problems}, Journal of Algorithms, 17 (1994), pp.~251--268.

\bibitem{mandal2006maximum}
{\sc S.~Mandal and M.~Pal}, {\em Maximum weight independent set of circular-arc
  graph and its application}, Journal of Applied Mathematics and Computing, 22
  (2006), pp.~161--174.

\bibitem{mcconnell2003linear}
{\sc R.~M. McConnell}, {\em Linear-time recognition of circular-arc graphs},
  Algorithmica, 37 (2003), pp.~93--147.

\bibitem{McKee1999}
{\sc T.~A. McKee and F.~R. McMorris}, {\em Topics in Intersection Graph
  Theory}, Society for Industrial and Applied Mathematics, Jan. 1999,
  \url{https://doi.org/10.1137/1.9780898719802},
  \url{https://doi.org/10.1137/1.9780898719802}.

\bibitem{NaorPY20}
{\sc M.~Naor, M.~Parter, and E.~Yogev}, {\em The power of distributed verifiers
  in interactive proofs}, in 31st ACM-SIAM Symposium on Discrete Algorithms
  (SODA), {SIAM}, 2020, pp.~1096--115.

\bibitem{Peleg00}
{\sc D.~Peleg}, {\em Distributed Computing: A Locality-Sensitive Approach},
  SIAM, 2000.

\bibitem{roberts1969indifference}
{\sc F.~S. Roberts}, {\em Indifference graphs. proof techniques in graph
  theory}, in Proceedings of the Second Ann Arbor Graph Conference, Academic
  Press, New York, 1969.

\bibitem{rose1976algorithmic}
{\sc D.~J. Rose, R.~E. Tarjan, and G.~S. Lueker}, {\em Algorithmic aspects of
  vertex elimination on graphs}, SIAM Journal on computing, 5 (1976),
  pp.~266--283.

\bibitem{tucker1970}
{\sc A.~Tucker}, {\em Characterizing circular-arc graphs}, Bulletin of the
  American Mathematical Society, 76 (1970), pp.~1257--1260.

\bibitem{yamazaki2020enumeration}
{\sc K.~Yamazaki, T.~Saitoh, M.~Kiyomi, and R.~Uehara}, {\em Enumeration of
  nonisomorphic interval graphs and nonisomorphic permutation graphs},
  Theoretical Computer Science, 806 (2020), pp.~310--322.

\end{thebibliography}

\end{document}